\pdfoutput=1
\documentclass[12pt, draftclsnofoot, onecolumn]{IEEEtran}
\usepackage{comment}
\usepackage{algorithm}
\usepackage{algpseudocode}
\usepackage{amsmath}
\usepackage{mathtools}
\usepackage{amsmath}
\usepackage{amssymb}
\usepackage{cite}
\usepackage{array}
\usepackage{amsthm}
\usepackage{amsmath,dsfont}
\usepackage{tabulary}
\usepackage{multirow}
\usepackage{cleveref}
\usepackage[subrefformat=parens,labelformat=parens]{subfig}
\usepackage{color}
\usepackage{pifont}
\usepackage{epsfig}
\usepackage{graphicx}
\usepackage{url}
\usepackage{amssymb}
\IEEEoverridecommandlockouts
\graphicspath{ {./figs/} }

\newcommand{\power}[1]{\mathrm{P}_{#1}}

\newtheorem{thm}{{\bf Theorem}}
\newtheorem{cor}{Corollary}

\theoremstyle{definition}
\newtheorem{definition}{Definition}

\begin{document}

%----------------------------------------------------------------------
% Title Information, Abstract and Keywords
%----------------------------------------------------------------------
\title{Max-Min Rates in Self-backhauled Millimeter Wave Cellular Networks}
\author{Mandar N. Kulkarni, Amitava Ghosh, and Jeffrey G. Andrews 
\thanks{Email:
$\mathtt{mandar.kulkarni@utexas.edu,jandrews@ece.utexas.edu,amitava.ghosh@nokia-bell-labs.com}$. M. Kulkarni and J. Andrews are with the University of Texas at Austin, TX. A. Ghosh is with Nokia Bell Labs, Naperville, IL. Last revised on \today. }}
\maketitle
% potential alternate title: Max-min rates in square lattice deployment of self-backhauled millimeter wave cellular networks
\begin{abstract}
This paper considers the following question for viable wide-area millimeter wave cellular networks. What is the maximum extended coverage area of a single fiber site using multi-hop relaying, while achieving a minimum target per user data rate? We formulate an optimization problem to maximize the minimum end-to-end per user data rate, and exploit unique features of millimeter wave deployments to yield a tractable solution.  The mesh network is modeled as a $k-$ring urban-canyon type deployment, where $k$ is the number of hops back to the fiber site. The total number of relays per fiber site grows as $k^2$.  We consider both integrated access-backhaul (IAB) and orthogonal access-backhaul (OAB) resource allocation strategies, as well as both half and full duplex base stations (BSs).  With a few validated simplifications, our results are given as simple closed-form expressions that are easy to evaluate even for large networks. Several design guidelines are provided, including on the choice of routing and scheduling strategy, the maximum allowable self-interference in full duplex relays and role of dual connectivity to reduce load imbalance across BSs.  For example, we show that for certain load conditions  there is very little gain to IAB (as considered for 5G) as opposed to tunable OAB (using separate spectrum for access and backhaul links); the latter being significantly simpler to implement.
\end{abstract}

\section{Introduction}
In order to deploy affordable millimeter wave (mmWave) cellular networks that can cover a large urban area, it is highly desirable to deploy self-backhauled networks\cite{PiKha,Qual16,Guttman17},  wherein a fraction of the base stations (BSs) have fiber-like backhaul and the rest of the BSs backhaul to the fiber sites wirelessly.  This introduces a tradeoff between deployment cost and end user rate, which decreases as the amount of multi-hop relaying increases.  Although such a mesh network architecture has been considered both theoretically and in practice many times in the past (as discussed next), with limited success, a few novel features of urban mmWave cellular systems lends to significant simplification.  In particular, the highly directional transmissions, strong blocking from buildings, and limited diffraction around corners\cite{Rap13,Kart17,BaiHea14,AndBai16} --  combined with an urban topography -- allow us to plausibly model the network as a noise-limited $k-$ring deployment model, as shown in Fig.~\ref{fig:SystemModel}, with BSs deployed on a 2-D square grid. The number of relays grows as $k^2$ with a fixed inter-site distance (ISD). We consider a single fiber site, ignoring edge effects, which maybe anyway negligible due to the noise-limitedness. This model will allow us to succinctly quantify the maximum rate achievable by all users, called max-min rates, in closed form. 

We focus on max-min rates for two reasons. The first is that it allows us to determine the maximum value of $k$, that is how far the mesh network can extend from the fiber site, while ensuring a certain end-to-end (e2e) per user data rate.  The second is that it results in a tractable optimization problem, as opposed to focusing on say, the 5th percentile user.  We provide several validations of the proposed model and results.  Given these tractable results, we consider three additional design choices, namely (i) integrated access-backhaul (IAB) or orthogonal access-backhaul (OAB) resource allocation, (ii) full or half duplex relays, and (iii) does dual connectivity improve per user rates in a self-backhauled network?  IAB allows access and backhaul links to share time-frequency resources, whereas OAB reserves different set of resources for access and backhaul links.

\subsection{Background, Motivation, and Related Work}
The study of multi-hop wireless networks has a rich history spanning theoretically optimal resource allocation schemes\cite{Haj88,Tassiulas92,Kelly1998,Jain2005}, scaling laws\cite{Gupta00,Franceschetti07} and analysis of achievable e2e metrics\cite{Gitman75,Kleinrock87,Sik06}. There has also been industry-driven standardization activities for multi-hop wireless local area networks (WLANs)\cite{80211s} and for fourth generation (4G) cellular networks with a single wireless backhaul hop\cite{3gpprelay,Peters09}. Practical implementation of multi-hop networks, however, has not been very successful. Reasons include the coupled interference and scheduling between hops \cite{Sik06},  large overheads for maintaining multi-hop routes, a lack of Shannon-like theoretical limits and their corresponding design guidance \cite{And08} and the fundamentally poor e2e-rate scaling caused by each packet having to be transmitted multiple times \cite {Gupta00}.

A key differentiating factor for mmWave cellular networks is that they can be designed to be noise-limited, especially with large bandwidths, small antenna beam widths \cite{SinJSAC14,Mud09} and appropriate resource allocation strategies. Noise-limitedness greatly simplifies the routing and multi-hop scheduling problems. Thus, identifying and studying network scenarios that lead to noise-limited behavior can help close the gap between theory and practice. Recently, \cite{Yuan18} proposes a polynomial time algorithm for joint routing and scheduling, extending the work in \cite{Haj88}, unlike traditional NP-hard solutions\cite{Tassiulas92,Jain2005}. However, \cite{Yuan18} considers a generic deployment topology and exploiting specific deployment patterns may result in even simpler structural solutions for optimal routing and scheduling, which could help arriving at the final formulae for e2e rates in closed form. %In this paper we show that although ignoring interference may not be desirable in developing optimal resource allocation strategies, in order to understand what is the max-min rate achievable for a given user deployment scenario and value of $k$, it suffices to consider the noise-limitedness assumption. 

The new-found interest in multi-hopping for mmWave is reflected in recent academic work such as \cite{Rois15,Nazmul17,He15,Du17,Ras15,LiCaire17,Yuan18,SinJSAC14,KulTWC17,Semiari17,Yan18}, as well as in 3GPP-Release 16 standardization activities\cite{Guttman17}. A cross layer optimization framework was proposed in \cite{Rois15}. In \cite{He15}, a fixed demand per flow traffic model was assumed to solve the problem of minimizing the time to empty the demands of all flows. In \cite{Nazmul17}, a joint cost minimization along with resource allocation optimization problem was formulated. In \cite{SinJSAC14,KulTWC17,Wu18}, per user rate analysis in mmWave self-backhauled networks was done using stochastic geometry considering a single backhaul hop. These frameworks do not trivially extend to analysis of multi-hop backhauling. Note that in most of the prior works which attempted to optimize resources in multi-hop mmWave networks the optimal solutions are NP-hard or require implementing a linear program that involve large matrices with the size growing very fast with the number of nodes in the network\cite{Yuan18,Du17}, although approximate simpler solutions have been attempted\cite{LiCaire17,He15,Yuan18,Semiari17}.  

The closest recent work to this work is \cite{Ras15}, which considers a general graph for deployment and includes out of cell interference leading to a linear programming solution for joint routing and scheduling. This is where our work differs, as we use noise-limitedness to simplify the problem formulation and exactly solve the optimization problem under consideration to give closed form results for maximizing minimum rate in $k-$ring networks considering different design choices and also provide structural results on optimal routing and scheduling. A summary of contributions of this manuscript is given next. %There is no doubt that the last 3 decades has seen several seminal contributions for bringing sophestication in modeling and analysis of multi-hop networks in terms modeling traffic types or modeling the received signal power, but we focus on a simple model in this work as it suffices to give a first order understanding of the titular question without compromising the dependence on some key network parameters.  This work should be viewed as a reference for providing simple rules of thumb for addressing the design of multi-hop self-backhauled mmWave cellular networks, which is fundamentally a very hard problem given the complex underlying physical and network layer phenomenon. 

\subsection{Contributions}
{\bf Closed form results for max-min optimal rates.}
We study a grid deployment of BSs with a single fiber site and $k^2$ relays around it, which we term a $k-$ring deployment. Arbitrary but static user equipment (UE) deployment is assumed with full buffer traffic. Downlink (DL) UEs are considered, unless specified otherwise. We compute closed form expressions of maximum e2e rate achievable by all UEs when the BSs are half or full duplex, and when IAB or OAB resource allocation strategy is used.  All rates are assumed to be deterministic, although our results for OAB are extendable when access rates are random but backhaul rates are deterministic. For general load scenarios, we compute the max-min rates given a routing strategy (that is optimizing only over scheduling). For arbitrary load scenario with equal backhaul rate per relay considering OAB, and for a special case of load scenario with IAB, we show it is possible to jointly optimize routing and scheduling to give closed form expressions for max-min rates. 

{\bf Applications of the analysis.}
We use the derived formulae in order to develop an understanding of the following question: what is the maximum value of $k$ such that all users in the network support e2e rates greater than a target threshold in the $k-$ring networks? We answer the question under several realistic network parameters. The max-min rates derived are also used to compare IAB versus OAB, single versus dual connectivity and half versus full duplex base stations, which also lead to interesting insights detailed in Section~\ref{sec:results}. For instance, in certain load conditions it can be possible to closely follow the max-min rates with IAB using an OAB scheme which can be simpler to implement in practice.

{\bf Validation of noise-limitedness.}
Considering a worst case interference model, we show that although ignoring interference may not be a good assumption for deriving optimal resource allocation strategies (in particular for determining routing/scheduling), deriving closed form expressions for max-min rates considering noise-limitedness suffices to give an achieavable limit for max-min rates in self-backhauled networks with few bottleneck links. This is because the optimal strategies chose only few active links at a time leading to noise-limited performance. Furthermore making use of some structural properties of the considered networks and a more realistic path loss model specific to urban canyon settings in \cite{Kart17,Wang18}, we show that noise-limitedness in fact holds in the self-backhauled networks considered in this work. 
\begin{figure}
\centering
{\includegraphics[width= 0.6\columnwidth]{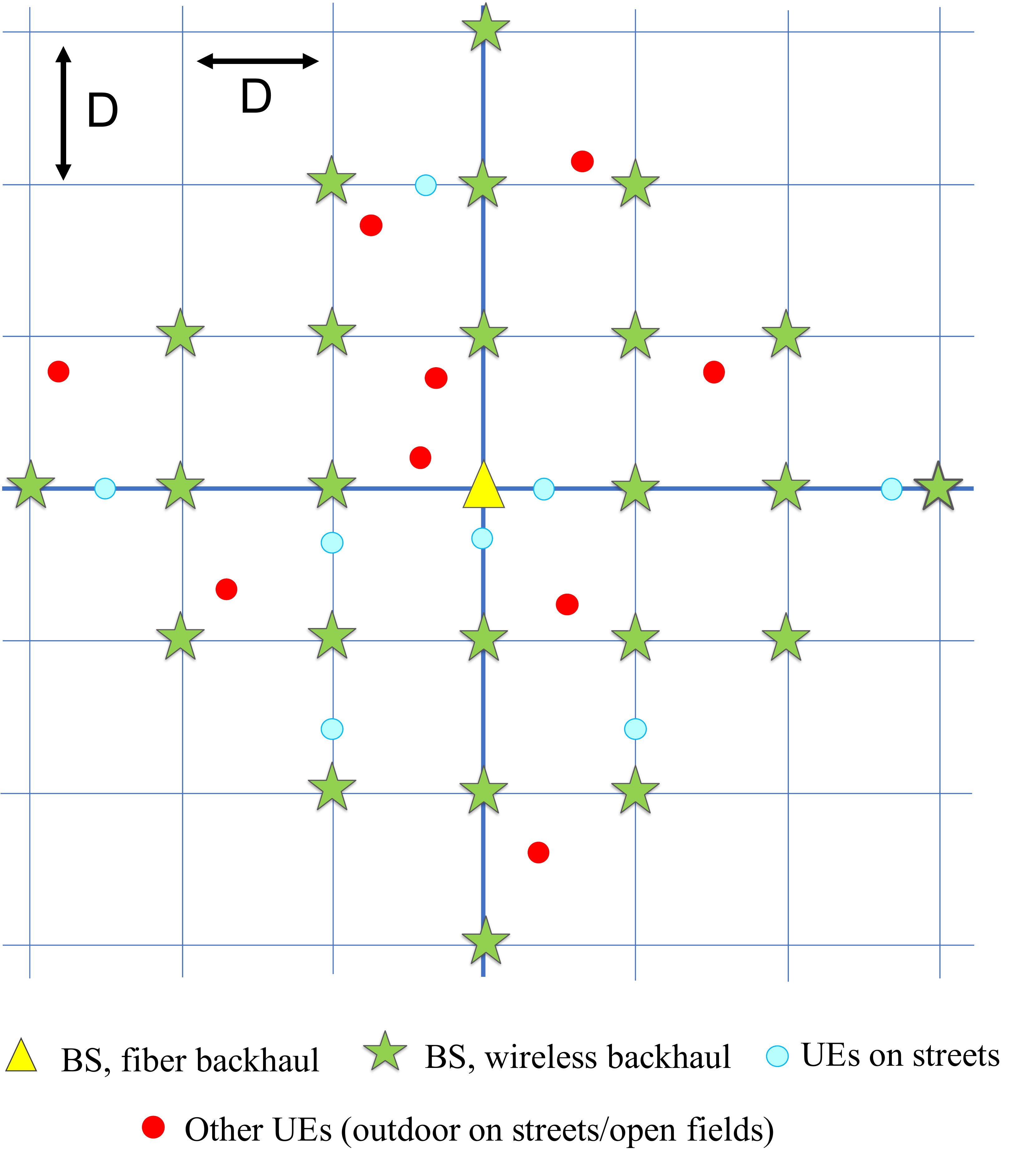}}
\caption{$k-$ring model, $k=3$.}
\label{fig:SystemModel}
\end{figure} 

\section{System Model}
\label{sec:deployment}
{\bf $k-$ring deployment.} We study a $k-$ring deployment model for urban canyon scenarios, as shown in Fig.~\ref{fig:SystemModel}. Lines represent streets on which BSs are deployed, with BSs denoted by either a triangle (MBS) or star (relays). The inter-line spacing is $\mathrm{D}$ meters. The MBS, which is the fiber backhauled BS or master BS, is located at $(0,0)$ and the relays are located at $(i\mathrm{D},j\mathrm{D})$ for $i,j\in\{0,\pm 1,\pm 2,\ldots, \pm k\}$ such that the Manhattan distance from any relay to the MBS is $\leq k\mathrm{D}$. We denote $(i\mathrm{D},j\mathrm{D})$ by $(i,j)$. All possible links (directed line joining any two nodes, which can be BSs or UEs) in the network are wireless. Arbitrary but fixed user deployment is assumed. Performance is evaluated for a static realization of UE locations, motivated from fixed wireless to home or other low mobility applications. 

DL flow is assumed unless specified otherwise, that is data flows for each UE originate at the fiber site and then end at the UEs after visiting zero or more relays. Time is continuous and total time is 1 unit. $\mathrm{U}$ is the total number of UEs in the network. Each UE associates with one BS according to any static association criterion which does not change with time. For example, nearest neighbour or minimum path loss association. UEs can only connect with BSs along the same street since path loss on links across orthogonal streets can be very high\cite{Wang18,Kart17}. Number of users connected to a BS at $(i,j)$ is denoted as $w_{i,j}$. All devices in the network are assumed to be half-duplex unless specified otherwise. BSs are assumed to transmit to or receive from only one device (UE/BS) at a time. No broadcast/multicast/network-coding traffic is allowed.
%Here, we give some examples of UE deployments to give an idea of the variety of scenarios that may be studied using this model. UEs can be either on the lines or in the squares formed by intersection of the lines. UEs in the squares can be outdoor UEs, in which case they lie on streets different from the ones on which the BSs are deployed, or some UEs maybe LOS outdoor in open fields inside the city. 
%\begin{figure}
%\centering
%{\includegraphics[width= 0.5\columnwidth]{SystemModel2}}
%\caption{$k-$ring model overlaid in an urban area with $\mathrm{D} = 100$m.}
%\label{fig:sysmodel2}
%\end{figure}

{\bf Routing and traffic model.} An ordered list of all nodes visited by a UE's data starting from the MBS is called the {\em route} of that UE. The route includes the UE itself. Every BS should be visited at most once in a route. For any node, say $(i,j)$, on a route, its preceding node (if it exists) is called the {\em parent} of $(i,j)$ on that route. Similarly, a succeeding node (if it exists) is called a {\em child} node of $(i,j)$ on that route. A {\em hop} on the route of a UE is a link between adjacent nodes in the route of the UE. It is assumed that {\em backhauling (that is BS to BS hops) can happen on links only along a street}\cite{Kart17,Wang18}. Furthermore, it is assumed there is a {\em unique} route from the fiber site to every UE and the routes are static (do not change with time). Full buffer traffic model is assumed. This implies that given a route of a UE, every BS along the route always has the UE's data to transmit. {\em Routing strategy} is defined as the collection of routes of all UEs. Given a routing strategy, $f(i,j)$ denotes the effective number of UEs served by $(i,j)$. That is, 
$f(i,j) = \sum_{u=1}^{\mathrm{U}}\mathbf{1}\{(i,j)\in\text{route of user }u\}.$
Note that $f(0,0) = \mathrm{U}$. 

{\bf Instantaneous rate and noise-limitedness.} Every link (access and backhaul) in the network is associated with a fixed number called instantaneous rate. If a link with instantaneous rate $R$ is activated for time $\tau$, then $\tau R$ is the data transmitted on that link. Let $R_i$ denote the deterministic instantaneous rate on a backhaul link of length $i\mathrm{D}$ for $i =  1, \ldots, k$. It is assumed that $R_i$ is decreasing with $i$. Assumptions on instantaneous access rates will be made in the next section. Note that backhaul links along a street will generally be LOS. Since LOS mmWave links have negligible small scale fading\cite{Rap13}, an assumption of deterministic instantaneous rates is justifiable. The analytical results with OAB are extended for random instantaneous rates for access links, which can incorporate the impact of dynamic blockages, and we will discuss more about this later. Another implicit but important assumption was made above. That is, the instantaneous rates are independent of the {\em transmission schedules}, that is the set of links activated simultaneously. This is essentially noise-limitedness assumption, which we will extensively validate in Section~\ref{sec:validation}. 

{\bf Scheduling assumptions.} Let $\mathrm{L}$ be the total number of links (including access and backhaul). {\em Feasible schedules} are defined by a collection of $\mathrm{L}\times \mathrm{U}$ matrices, called {\em scheduling matrices}, which are described next. Each entry, $\tau_{l,u}$, in a scheduling matrix indicates fraction of time link $l$ was used to serve data for user $u$. Here, $0 \leq \tau_{l,u}\leq 1$. Furthermore, since the total time over which optimization is done is 1 unit, the fraction of time every BS is active (that is either transmitting or receiving) is at most 1. That is, 
$\sum_{l\in\mathcal{L}_{i,j}}\sum_{u=1}^{\mathrm{U}}\tau_{l,u}\leq 1,$ where $\mathcal{L}_{i,j}$ denotes the set of links with $(i,j)$ as one of the endpoints. There should also be an achievable algorithm, called {\em link activation strategy}, for activating  link $l$ in the network for a total $\sum_{u=1}^{\mathrm{U}}\tau_{l,u}$ fraction of time, $\forall l$, such that every BS transmits to or receives from only one device (UE or BS) at a time. A link activation strategy can also be defined in terms of a collection of transmission schedules (a subset of $\{1,\ldots,\mathrm{L}\}$) and their corresponding activation times. Let there be $\mathrm{P}$ transmission schedules that satisfy single active link per BS constraint.
Let $0\leq x_{p} \leq 1$ denote the fraction of time transmission schedule $p$ was activated (where $p=1,\ldots, \mathrm{P}$). A link activation strategy  $\{x_p\}_{p=1,\ldots,\mathrm{P}}$ is said to be compatible with a scheduling matrix with entries $\tau_{l,u}$ if $\sum_{u=1}^{\mathrm{U}}\tau_{l,u} = \sum_{p=1}^{\mathrm{P}}\mathds{1}(\text{link }l\in\mathrm{P})x_{\mathrm{P}}$.
Given a routing strategy, there is an extra constraint on the scheduling matrix as follows. If link $l$ is not a hop on route of UE $u$, $\tau_{l,u} = 0,$ $\forall u\in\{1,\ldots, \mathrm{U}\}$.

\section{Max-min end to end rate in $k-$ring deployment}
\label{sec:maxmin}
We want to find what is the maximum value of $k$ that can support a target e2e rate achieved by all UEs. We instead fix a $k$ and find the maximum e2e rate achieved by all UEs. Let us define this formally. {\em Long term rate} of a user $u$ on link $l$ is defined as $\tau R$, where $R$ is the instantaneous rate on link $l$ and $\tau$ ($<1$) is the fraction of time that user $u$ was scheduled on link $l$. Given a routing strategy and a corresponding scheduling matrix $\mathcal{S}$, the {\em e2e rate of UE} $u$, denoted as $R^\mathcal{S}_{u}$, is the minimum of its long term rate over all hops from the fiber site to the UE. 

\begin{definition}
\label{defn: maxmin}
Given a routing strategy, the max-min rate is defined as $\gamma^* = \max_{\mathcal{S}}\theta_\mathcal{S},$ where $\theta_\mathcal{S} = \min_{u=1,\ldots, \mathrm{U}} R^\mathcal{S}_{u}$ with $R^\mathcal{S}_{u}$ being the e2e rate of user $u$ for scheduling matrix $\mathcal{S}$. Maximizing $\gamma^*$ considering all routing strategies that are feasible as per the system model in Section~\ref{sec:deployment}, we obtain the globally optimal max-min rate denoted as $R^*_\mathrm{e2e}$.
\end{definition}

%IAB allows access and backhaul to share time-frequency resources, whereas there is a fixed split between access and backhaul across all BSs with OAB. %Thus, we expect IAB rates to be at least as good as OAB rates owing to noise-limitedness. 
\subsection{Integrated access backhaul}
\label{sec:IAB}
Given a routing strategy, an upper bound to the max-min rate optimization problem can be found by solving \eqref{optformulation}, where $r_l$ is the instantaneous rate on link $l$ and $\mathcal{S}$ denotes $\mathrm{L}\times \mathrm{U}$ matrix with elements $\tau_{l,u}$ for $l=1,\ldots,\mathrm{L}$ and $u=1,\ldots,\mathrm{U}$. Since the values of $\tau_{l,u}$ that solve \eqref{optformulation} need not satisfy the constraint that every BS scheduled at most one link at a time instance, \eqref{optformulation} gives an upper bound to max-min rate and does not necessarily give an achievable value.
For IAB, the solution of \eqref{optformulation} is equal to the max-min rate $\gamma^*$, if there exists a link activation strategy which is compatible with the  scheduling matrix which maximizes $\theta_\mathcal{S}$ in \eqref{optformulation}. 
\begin{figure*}
\begin{equation}
\label{optformulation}
\begin{aligned}
& \underset{\mathcal{S}}{\text{maximize}}
& & \theta_\mathcal{S} \\
& \text{subject to}
& & \theta_\mathcal{S} \leq r_{l}\tau_{l,u}, \forall\; \text{hop }
 l \text{ on the route of user }u, \forall u=1,\ldots, \mathrm{U}\\
& & &\sum_{l\in\mathcal{L}_{i,j}}\sum_{u=1}^{\mathrm{U}}\tau_{l,u}\leq 1, \forall\; i,j\in\{0,\pm 1,\ldots, \pm k\} \text{ s.t. }|i|+|j|\leq k\\
& & & 0 \leq \tau_{l,u}\leq 1, \forall\; l\in\{1,2,\ldots, \mathrm{L}\}, \forall u\in\{1,\ldots, \mathrm{U}\}\\
& & & \tau_{l,u} = 0, \text{ if link }l \text{ is not a hop on route of UE }u, \forall u\in\{1,\ldots, \mathrm{U}\}
\end{aligned}
\end{equation}
\end{figure*}

For simplicity of exposition, let us number the BSs in the network from $0$ to $2k(k+1)$. BS index $0$ corresponds to the MBS. With some abuse of notation $f(i)$ denotes effective load on BS $i$ (number of UEs with the BS $i$ on it's route) with new indexing under some routing strategy. Let us also number the users from $1$ to $f(0)$ in ascending order of the index of their corresponding serving BS. $R_{a,u}$ denotes instantaneous access rates of users for $u=1,\ldots, f(0)$.
\begin{thm}
\label{thm:IAB2}
Given any nearest neighbour routing strategy such that every BS has a unique parent node, 
\begin{equation}
\label{eq:IAB2}
\gamma^* = \max_{\mathcal{S}}\min_{u}R^\mathcal{S}_u = \left(\max_{i\in\{0,\ldots, 2k(k+1)\}}{\bf c}_i^{T}{\bf b}\right)^{-1},
\end{equation}
where $\mathcal{S}$ is denotes set of all feasible scheduling matrices given the routing strategy and $u$ denotes all the users in the network. Here, $${\bf c}_i = \mathrm{sum}\left({\bf I}_{f(0)+1}, \sum_{r=0}^{i-1}w_r, \sum_{t=0}^{i}w_t \right)  + \left(2f(i)-w_i\right){\bf e}_{f(0)+1}, \forall i\neq 0,$$
$${\bf c}_0 = \mathrm{sum}\left({\bf I}_{f(0)+1}, 0, w_0 \right)  + \left(f(0)-w_0\right){\bf e}_{f(0)+1},$$
where ${\bf e}_j$ represents the $j^{\text{th}}$ column of identity matrix of dimension $f(0)+1$ and $\mathrm{sum}\left({\bf I}_{f(0)+1},l,u\right) = \sum_{j=l+1}^{u}{\bf e}_j$. Here, ${\bf b} = \left[\frac{1}{R_{a,0}}\ \frac{1}{R_{a,1}}\ \ldots\ \frac{1}{R_{a,f(0)}}\ \frac{1}{R_1}\right]^{T},$ where $R_{a,u}$ is the access rate to $u^\text{th}$ user.  
\end{thm}
\begin{proof}
If $\gamma$ is a minimum rate achieved by all users, then for BS with index $i$ the following inequality should be satisfied.
$$\gamma\left(\sum_{u=l_i+1}^{l_i+w_i}\frac{1}{R_{a,u}} + \frac{f(i)-w_i}{R_1} + \frac{\mathds{1}(i\neq 0)f(i)}{R_1}\right)\leq 1, \forall i = 0,\ldots, 2k(k+1),$$
where $l_i = \sum_{r=0}^{i-1}w_r$ is chosen such that the indices from $l_i+1$ to $l_i+w_i$ correspond to UEs associated with BS $i$. This inequality can be interpreted as follows. Here, $\gamma \sum_{u=l_i+1}^{l_i+w_i}\frac{1}{R_{a,u}} $ is the minimum fraction of time BS $i$ spends on access links. Next, $\gamma \frac{f(i)-w_i}{R_1} $ is the minimum fraction of time the BS spends on transmitting data to its children nodes on all the routes which pass through the BS. Finally, $\gamma\frac{\mathds{1}(i\neq 0)f(i)}{R_1}$ is the minimum fraction of time the BS spends on receiving data from it's parent node. This leads to the upper bound on max-min rate given by $\gamma^* = \frac{1}{\max_i {\bf c}^T_i {\bf b}}$. To prove the upper bound is achievable, the following scheduler  is sufficient. The MBS allocates $\frac{(f(0)-w_0)\gamma^*}{R_1}$ fraction of time for backhaul and rest for access. MBS equally divides the backhaul time amongst the $f(0)-w_0$ users which are connected to the MBS through relays. The MBS allocates $\gamma^*/R_{a,u}$ fraction of time for user $u$ directly connected to the MBS, where $u = 1,\ldots, w_{0}$. A ring 1 relay (say $i$) allocates $\frac{(f(i)-w_i)\gamma^*}{R_1}$ fraction of time for serving its children relays considering a DL flow, when it is not scheduled by the MBS for receiving backhaul data. The ring 1 relay equally divides its backhaul time amongst the $f(i)-w_i$ users indirectly connected to it through relays. The ring 1 relay further allocates $\gamma^*/R_{a,u}$ fraction of time for user $u$ connected to it, where $u=1,\ldots, w_i$. This process continues hierarchically for all relays in the $k-$ring deployment. All the constraints in \eqref{optformulation} are satisfied by the achieavable algorithm, while ensuring that the e2e rate of each UE satisfies $\gamma^* \left(\sum_{u=l_i+1}^{l_i+w_i}\frac{1}{R_{a,u}} + \frac{f(i)-w_i}{R_1} + \frac{\mathds{1}(i\neq 0)f(i)}{R_1}\right)\leq 1, \forall i = 0,\ldots, 2k(k+1)$. Furthermore, the algorithm ensures that the link activation strategy has at most one active link per BS at any instant of time. 
\end{proof}

\begin{cor}
\label{cor:IAB2}
Given any static routing strategy -- need not be nearest neighbour and different UEs connected to a BS can have different routes -- $\left(\max_{i\in\{0,\ldots, 2k(k+1)\}}{\bf c}_i^{T}{\bf b}\right)^{-1}$ in \eqref{eq:IAB2} gives an upper bound on max-min rate. 
\end{cor}
\begin{proof}
The proof follows that of Theorem~\ref{thm:IAB2}. The inequalities $$\gamma\left(\sum_{u=l_i+1}^{l_i+w_i}\frac{1}{R_{a,u}} + \frac{f(i)-w_i}{R_1} + \frac{\mathds{1}(i\neq 0)f(i)}{R_1}\right)\leq 1, \forall i = 0,\ldots, 2k(k+1),$$
remain unchanged inspite of different routing strategies being considered. This is because $R_1>R_i$, for all $i>1$. Thus, considering rate $R_1$, which corresponds to nearest neighbour backhaul hop, in above inequalities still leads to an upper bound on the max-min rate. 
\end{proof}

Note that the values of effective load $f(i)$ can be different for different routing strategies. Thus, the above corollary does not imply optimality of nearest neighbour routing, which could be an impression the reader might get given that Theorem~\ref{thm:IAB2} gives exact max-min rates and Corollary~\ref{cor:IAB2} gives an upper bound on max-min rates. However, one could evauate the upper bound in Corollary~\ref{cor:IAB2} for different routing strategies and compare it with the max-min rate obtained using Theorem~\ref{thm:IAB2} to numerically evaluate whether nearest neighbour routing strategies perform near optimally. 
Now, we extend the result in Theorem~\ref{thm:IAB2} considering full duplex BSs. Since there will be self-interference at each relay the access rates and backhaul rates will be different than in Theorem~\ref{thm:IAB2}. Let access rate for user $u$ under full duplex relaying be $R^{f}_{a,u}\leq R_{a,u}$ and the single hop backhaul rate be $R^f_1 \leq R_1$. Although the system model set up in Section~\ref{sec:deployment} was for DL, all the relevant definitions can be extended for UL as well. We now consider a scenario when there are some UL and some DL UEs. Note that a UE cannot be both UL and DL. Let $\mathcal{U}_\mathrm{DL}$ and $\mathcal{U}_{\mathrm{UL}}$ be the set of indices of downlink (DL) and uplink (UL) UEs. Let $w^{\mathrm{DL}}_{i}$ and $w^{\mathrm{UL}}_{i}$ denote the number of DL and UL UEs connected to BS $i$. Similarly,  $f^{\mathrm{DL}}(i)$ and $f^{\mathrm{UL}}(i)$ corresponds to effective DL and UL load on BS indexed by $i$. 

\begin{thm}
\label{thm:IAB3}
Assuming all BSs operate in full duplex mode and given any nearest neighbour routing strategy  with unique parent node for every BS, $\gamma^* = \max_{\mathcal{S}}\min_{u}R^\mathcal{S}_u = \min\left(\gamma_{\mathrm{tx}},\gamma_{\mathrm{rx}}\right)$, where $\mathcal{S}$ is denotes set of all feasible scheduling matrices given the routing strategy and $u$ denotes all the users in the network. Here, $\gamma_{\mathrm{tx}} = \left(\max_{i\in\{0,\ldots, 2k(k+1)\}}{\bf c}_{\mathrm{tx},i}^{T}{\bf b}_f\right)^{-1}$ and $\gamma_{\mathrm{rx}} = \left(\max_{i\in\{0,\ldots, 2k(k+1)\}}{\bf c}_{\mathrm{rx},i}^{T}{\bf b}_f\right)^{-1}$. Also, 
$${\bf b}_f = \left[\frac{1}{R^f_{a,0}}\ \frac{1}{R^f_{a,1}}\ \ldots\ \frac{1}{R^f_{a,f(0)}}\ \frac{1}{R^f_1}\right]^{T},$$
$${\bf c}_{\mathrm{tx},i} = \mathrm{sum}^{\mathrm{DL}}\left({\bf I}_{f(0)+1}, \sum_{k=0}^{i-1}w_k, \sum_{k=0}^{i}w_k \right)  + \left(f(i)-w^\mathrm{DL}_i\right){\bf e}_{f(0)+1}, \;\forall i\neq 0,$$
$${\bf c}_{\mathrm{tx},0} = \mathrm{sum}^{\mathrm{DL}}\left({\bf I}_{f(0)+1}, 0, w_0 \right)  + \left(f^{\mathrm{DL}}(0)-w^{\mathrm{DL}}_0\right){\bf e}_{f(0)+1},$$
$$\mathrm{sum}^{\mathrm{DL}}({\bf I}_{f(0)+1}, l, u) = \sum_{j=l+1}^{u}{\bf e}_j\mathds{1}\left(\text{\em UE $j$ is DL}\right),$$
where ${\bf e}_j$ represents the $j^{\text{th}}$ column of identity matrix of dimension $f(0)+1$. ${\bf c}_{\mathrm{rx},i}$ is same as ${\bf c}_{\mathrm{tx},i}$ but with superscript $\mathrm{DL}$ replaced by  $\mathrm{UL}$ in all places. 
\end{thm}
The proof of Theorem~\ref{thm:IAB3} is similar to that of Theorem ~\ref{thm:IAB2}.  The main difference is that instead of the required fraction of time for reception plus transmission being $\leq 1$ for every BS, we now have two separate inequalities per BS -- one for transmission time and one for reception time. Note that the analysis of e2e rates of UL UEs follows trivially from DL analysis given noise-limitedness assumption. Similar to Corollary~\ref{cor:IAB2}, the formula for max-min rate 
in the above theorem works as an upper bound for arbitrary static routing strategy. 

Until now we focused on specifying the max-min rates when a routing strategy is given. In the next result, we show that one can further theoretically characterize the optimal routing strategy as well under some simplistic assumptions on load and access rate to all UEs. In order to state we next result we define a class of routing strategies called {\em highway routing} next. 

Streets along the X and Y axes are called as {\em highways}.  All UEs associated with a BS at $(i,j)$ have same route from the fiber site to the associated BS. Under a highway routing strategy, the route from fiber site to $(i,j)$ is as follows. First the fiber site transmits data to either $(i,0)$ or $(0,j)$, whichever is furthest in terms of Manhattan distance, potentially over multiple hops along the shortest path joining the two nodes. From $(i,0)$ or $(0,j)$ the data is then transmitted to the $(i,j)$ along the shortest path in terms of Manhattan distance, potentially over multiple hops. The Manhattan distance of $(i,j)$ from the MBS decreases with every DL hop. If $|i|=|j|$, then the traffic is directed to either $(i,0)$ or $(0,j)$ but not both. However, if $(0,0)\to(i,0)\to(i,i)$ then $(0,0)\to(-i,0)\to(-i,-i)$.  If there were UL paths, then those would be exactly same as DL paths but in reverse order. Theorem~\ref{thm:IAB1} proves the optimality of nearest neighbour highway routing (NNHR) in specific load scenarios and when access rates to all users is the same. We then discuss why NNHR is a good choice in more general load and access rate settings. Note that Theorems~\ref{thm:IAB2} and \ref{thm:IAB3} already generalized the stringent conditions on load and access rates in Theorem~\ref{thm:IAB1} and do not rely on NNHR. We discuss Theorem~\ref{thm:IAB1} primarily for developing an intuitive understanding of when would NNHR routing strategy be theoretically optimal. In Section~\ref{sec:validation}-A, we empirically observe that the conditions in Theorem~\ref{thm:IAB1} may be further relaxed for optimality of NNHR . 

\begin{thm}
\label{thm:IAB1}
Let $w_{0,0}\geq w_{i,j}$ and $w_{i,j} = w_{-i,-j}$ $\forall i,j\in\{0,\pm1,\ldots, \pm k\}$. NNHR is optimal in terms of max-min rates and the optimal rate is given by $R^*_{\mathrm{e2e}} = \left(\frac{w_{0,0}}{R_a} + \frac{f(0,0)-w_{0,0}}{R_1}\right)^{-1}$. The formula for $R^*_{\mathrm{e2e}} $ gives an upper bound to max-min rate even if the load constraints $w_{0,0}\geq w_{i,j}$ and $w_{i,j} = w_{-i,-j}$ do not hold. 
\end{thm}
\begin{proof}
See Appendix~\ref{app:thmproof}. 
\end{proof}
Although the assumptions in Theorem~\ref{thm:IAB1} are idealistic, it gives an intuition that NNHR can be a good choice when the {\em bottleneck node} in the network is the fiber site and the effective load on the fiber site is well balanced in all four directions. A {\em bottleneck} node is formally defined as the node that has at least one link that is always active in order to attain the max-min rates. Also since the derived formula is simple, it offers a quick feasibility check for what is the maximum $k$ that supports a target rate. See Section~\ref{sec:softmaxmin} for a related discussion. 

NNHR may not be desirable in all possible load conditions. However, since having dynamic routing requires exchange of control signals and a more complex system design, it would be desirable to design a system wherein some static routing always gives a reasonable performance. In order to do this network planning, which includes deciding how many antennas should be employed at different BSs in the $k-$ring deployment or their transmit powers, can play an important role. If the BSs on the highways have much larger antenna gains than the non-highway relays then irrespective of the load it will be beneficial for the relays to employ the highway routing strategies since the highways links have much larger capacity to carry traffic than the non-highway links.

\subsection{Orthogonal access backhaul}
\label{sec:OAB}
Let $\zeta$ be the fraction of resources reserved for access and rest are reserved for backhaul. Every BS is assumed to divide the access time equally amongst all UEs directly associated with it. Thus, the long term access rate of each UE associated with BS $(i,j)$ is given as $\zeta R_a/w_{i,j}$. 

If a backhaul link with instantaneous rate $R$ is activated for $\tau$ fraction of time to serve all UEs associated with a relay, then {\em long term backhaul rate of a relay on a link} is defined as $\tau R$. Furthermore, {\em e2e backhaul rate} of a relay is defined as minimum of long term backhaul rate of the relay over each hop from the fiber site to the relay. 

We first consider a simple OAB scheme wherein equal e2e backhaul rate is offered to each relay. As this scheme does not optimize the rates based on load per BS, there will be some over-utilized and some under-utilized BSs. This issue, however, can be addressed by enabling dual connectivity which we will study in the next section. For simplicity of exposition, assume DL backhauling. The analysis holds for a mix of DL and UL backhauling since instantaneous link rates on backhaul do not change for UL and DL.
  
\begin{thm}
\label{thm:OAB1}
Maximum e2e backhaul rate that can be offered to each relay is given by $ \frac{(1-\zeta)R_1}{2k(k+1) w_{i,j}}$. NNHR is optimal for achieving this rate. Furthermore, e2e rate for any user connected to some BS at $(i,j)$ is given by $\frac{1}{w_{i,j}}\min(\zeta R_a, \frac{(1-\zeta)R_1}{2k(k+1)} )$, assuming OAB with $\zeta$ fraction of time for access and that every BS divides the access and e2e backhaul rates equally amongst the users directly associated with that BS.
\end{thm}
\begin{proof}
Let $\gamma$ be the e2e backhaul rate offered to each relay. Then the following should be satisfied. Under NNHR, $\gamma/R_1$ is the minimum fraction of resources the MBS allocates to each of the $2k(k+1)$ relays over the first backhaul hop. Thus, $\gamma (2k(k+1))/R_1\leq (1-\zeta)$. Let $g(i, j)-1$ represent total number of relays served by $(i,j)$. The following inequalities should also hold. $\gamma\left(\frac{g(i,j)-1}{R_1} + \frac{g(i,j)}{R_1}\right)\leq 1-\zeta$, for all $(i,j)\neq (0,0)$. Here, $\gamma g(i,j)/R_1$ is the fraction of time for relaying data to $(i,j)$ from its parent node. Also, $\gamma\frac{g(i,j)-1}{R_1}$ is the fraction of time for relaying data from $(i,j)$ to its children nodes. Since $2g(i,j)-1<2k(k+1)$, which holds because $g(i,j) = g(-i,-j)$ considering NNHR, the least upper bound on $\gamma$ is $\gamma\leq (1-\zeta)R_1/2k(k+1)$. This is achieved by using a scheduler similar to Algorithm~\ref{algo1}. The main modifications in Algorithm~\ref{algo1} are that $R_a$ is set to $\infty$, which makes time allocated for access equal to zero, total backhaul scheduling time is $1-\zeta$ and $w_{i,j}$ is set to 1 making $f(i,j) = g(i,j)$ in the description of Algorithm 1. 

A non-NNHR scheme cannot offer rates higher than $(1-\zeta)R_1/2k(k+1)$ as the inequality $\gamma (2k(k+1))/R_1\leq (1-\zeta)$ needs to hold irrespective of the routing strategies. Thus, $\gamma = (1-\zeta)R_1/2k(k+1)$ is the maximum e2e backhaul rate that can be offered to each relay. By definition, the e2e rate for a user is the minimum of its access long term rate and e2e backhaul rate. Consider a UE connected to $(i,j)$. Since backhaul rate to $(i,j)$ is equally divided amongst all $w_{i,j}$ users, the e2e backhaul rate of the UE is $\frac{(1-\zeta)R_1}{w_{i,j} 2k(k+1)}$. Long term access rate of the UE is $\zeta R_a/w_{i,j}$, since each user connected to a relay receives equal fraction of time for access. Thus, the e2e rate for the user is given by $\frac{1}{w_{i,j}}\min(\zeta R_a, \frac{(1-\zeta)R_1}{2k(k+1)} )$.
\end{proof}
\begin{cor}
\label{cor:OAB1}
If $w_{i,j} = w_{-i,-j}$ and $w_{0,0}\geq w_{i,j}$ and access rates to all UEs are given by $R_a$, there exists an OAB strategy that performs as good as IAB in terms of max-min rates. 
\end{cor}
Following are highlights of the proof of the above corollary.  Consider the following OAB scheme. $\zeta$ is the fraction of access resources (also called access frame) and $1-\zeta$ is the fraction of backhaul resources (also called backhaul frame). Within the backhaul frame, target long term rate to each relay is $\gamma w_{i,j}$, for all $i,j$. Routing and scheduling in backhaul frame is done to maximize $\gamma$.  It can be shown that the max-min rate achieved through this OAB scheme is same as the max-min rate derived in Theorem~\ref{thm:IAB1} as follows. Similar to the proof of Theorem~\ref{thm:IAB1}, the maximum achievable $\gamma$ for the OAB scheme under consideration is $\gamma = \frac{(1-\zeta)R_1}{f(0,0)-w_{0,0}}$ and NNHR is an optimal routing strategy. The key difference compared to proof of Theorem~\ref{thm:IAB1} is that since we are only optimizing routing and scheduling in backhaul resources, total time is limited to $1-\zeta$ instead of 1, and $R_a$ is set to infinity to make resources allocated for access equal to zero in the proof of Theorem~\ref{thm:IAB1}. Thus, with the OAB scheme under consideration the e2e rate for a user connected to a BS at $(i,j)$ is given by $\min( \frac{\zeta R_a}{w_{i,j}}, \frac{(1-\zeta)R_1}{f(0,0)-w_{0,0}} )$, assuming round robin scheduling done by $(i,j)$ amongst $w_{i,j}$ UEs for access and that the e2e backhaul rate to $(i,j)$ was equally divided amongst all $w_{i,j}$ UEs. Minimum e2e rate corresponds to $i=j=0$. Maximizing minimum e2e rate over $\zeta$, it is found that the max-min rate equals $\left(\frac{w_{0,0}}{R_a} + \frac{f(0,0)-w_{0,0}}{R_1}\right)^{-1}$, same as Theorem~\ref{thm:IAB1}.

%End-to-end backhaul rate with the OAB scheme described in the proof of Corollary~\ref{cor:OAB1} can be analyzed in general load scenario, like in Theorem~\ref{thm:IAB2}. In this case, the e2e backhaul rate is exactly same as that in Theorem~\ref{thm:IAB2} but replacing $1/R_{a,u}$ by $0$ in the definition of vector ${\bf b}$. 

The result in Theorem 4 can be extended considering random access rates. We now briefly discuss how this can be done. First, the definition of long term rate needs to be modified since the instantaneous access rates are no longer deterministic. Let the total scheduling time be $T$ units. The long term rate of a UE on a link is now defined as $\lim_{T\to\infty}\frac{1}{T}\int_{0}^{\tau T}X(t)\mathrm{d}t$, where $\tau$ is the fraction of time the link was active to serve the UE and $X(t)$ is a stationary ergodic random process and denotes the instantaneous rate of the link as a function of time. Using the ergodic theorem, the long term access rates is equal to $\tau \mathbb{E}\left[R_a\right]$, where $\tau$ is the fraction of time the access link is active. Considering deterministic backhaul rates, long term backhaul rate of a link is same as our earlier definition -- that is $\tau R$, where $\tau$ is the fraction of time backhaul link is active and $R$ is the instantaneous backhaul rate of the link. Since optimization in Theorem~4 is done only over routing and scheduling for backhaul links, the final result for e2e backhaul rate does not change. Since long term access rate is now $\zeta\mathbb{E}\left[R_a\right]/w_{i,j}$, the final e2e rate result in Theorem~4 remains unchanged with the exception that $R_a$ is replaced by $\mathbb{E}\left[R_a\right]$.

\section{Example Applications of the Analysis}
In this section, we discuss simple applications of our analysis. 
\subsection{5G Networks with Minimum Rate of 100 Mbps}
\label{sec:softmaxmin}
Deploying a new cellular network operating at mmWave involves significant cost and time overheads. Thus, it does not make sense if the deployed mmWave network offers only marginal gains over existing 4G networks. A minimum 100 Mbps per UE target has been set for 5G networks operating at mmWave frequencies. The analysis can be used to evaluate feasibility of potential BS or UE deployments for 5G networks. 
\subsubsection{Minimum number of rings required to get 100 Mbps rates}
A closed-form expression for maximum $k$ that supports 100 Mbps per UE can be obtained in simple settings like Theorem~\ref{thm:IAB1}. 
\begin{cor}
\label{cor:simpleformula}
The maximum $k$ that can still meet the max-min target rate of $\gamma_\mathrm{target}$ is given by 
$k\leq \frac{\sqrt{1+2R_1\left(\frac{1}{w\gamma_\mathrm{target}}-\frac{1}{R_a}\right)}-1}{2},$
if all relays have equal load $w$ and $\gamma_{\mathrm{target}}>\frac{R_a}{w}$.
\end{cor}
\begin{proof}
The max-min rate is given by $\gamma^* = \frac{1}{w}\left(\frac{1}{R_a} + \frac{2k(k+1)}{R_1}\right)^{-1}.$
Rearranging and solving the quadratic equation we get the result by using $\gamma^*\geq \gamma_\mathrm{target}$. 
\end{proof}
%Example. If $\gamma_\mathrm{target} = 100$ Mbps and $w = 5$ full buffered active UE per BS, access rates should be equal to at least $500$ Mbps to meet this for $k=0$. Now suppose $R_1 = 10$ Gbps and we need to design a network with $k=2$ rings. The required access rate is $R_a \geq  2.5$ Gbps. Thus, with a bandwidth of $1$ GHz this translates to a spectral efficiency of at least $10$ bps/Hz for backhaul links and a spectral efficiency of $2.5$ bps/Hz on access links. Using the physical layer models in Remark~1 and Remark~2, the following configuration can meet this requirement considering worst case NLOS UEs on the street at a distance of $50$m from the BSs, which are spaced on a grid with ISD$=100$m .  BSs have $64$ antennas, UEs have $16$ antennas, and transmit power is $1$W for all devices. These numbers are reasonable; 5G mmWave access points will have up to $1024$ antennas, with UEs having up to $64$ antennas\cite{Nokia16}. 

\subsubsection{Soft max-min}
Strictly maximizing the minimum rate in a mmWave system may lead to very poor e2e rates achieved by all UEs if a few of the UEs have very poor spectral efficiency, e.g. they are severely blocked by surrounding objects. Thus, it is practically beneficial to softly optimize the max-min rates. Here, we discuss a possible procedure. UEs that have very poor spectral efficiency, denoted as ``bad UEs", are placed with pseudo UEs for finding max-min rates. The pseudo UEs fake a higher signal to interference plus noise ratio (SINR) for the corresponding ``bad UEs". This allows the rest of the ``good UEs" to have much better rates after max-min optimization is performed. Essentially, these ``bad UEs" sacrifice themselves for the benefit of the whole. In a practical 5G system, such UEs would soon switch to a sub-6GHz legacy band to maintain a minimum performance level. 

\subsection{Analyzing performance of dual-connectivity.}
\label{sec:dualconn}
Multi-connectivity, wherein a UE connects to multiple BSs on the same or different bands, can counteract dynamic blocking in mmWave cellular. For self-backhauled networks, dual connectivity has another advantage to smooth out the load imbalance across all BSs. This can make resource allocation simpler in self-backhauled networks since employing equal rate per relay OAB is much simpler than IAB. 

Here, we look at a specific implementation of dual-connectivity (DC). OAB is assumed with $\zeta$ fraction of resources for access. DL UEs are assumed. Optimization is done to offer equal backhaul rates per relay. Consider a user connected to two BSs offering least path loss. Consider a user connected to relays at $(i,j)$ and $(i-1,j)$. Let the distance from the two BSs be $x$ and $y(<x)$, respectively. It is assumed that the UE has at least two RF chains so that it can receive signals from both connected BSs simultaneously. $R_a(x)$ is the access rate to the user from BS at $(i,j)$ and $R_a(y)$ is the access rate from BS at $(i-1,j)$. Let $R_{single}$ and $R_{dual}$ be the rates of the user under single connectivity (SC) and DC. Using Theorem~\ref{thm:OAB1}, $R_{single} = \frac{1}{w_{i,j}}\min\left(\zeta R_a(x), \frac{(1-\zeta)R_1}{2k(k+1)} \right)$. Considering DC,  $r_1 = \frac{1}{w^{'}_{i,j}}\min\left(\zeta R_a(x), \frac{(1-\zeta)R_1}{2k(k+1)} \right)$ and $r_2 = \frac{1}{w^{'}_{i-1,j}}\min\left(\zeta R_a(y), \frac{(1-\zeta)R_1}{2k(k+1)} \right)$ are e2e rates of the UE over its two connections, where $w^{'}_{i,j}$ ($\geq w_{i,j}$) is the new load on $(i,j)$ after dual connectivity. Thus, $R_{dual}$ can be defined as $R_{dual} = r_1 + r_2$, where additivity arises from the interpretation of e2e rate as total data transmitted from source to destination in 1 unit time. 

Following remarks describe how to compute the access and backhaul rates for evaluating the formulae derived in this paper. 

{\bf Remark 1} (Computing access rates).
$R_a(x) = \mathrm{W}\min\left(\log_2\left(1+\mathrm{SNR}_a\right),\mathrm{SE}_\mathrm{max}\right)$, where $\mathrm{SNR}_a$ is the effective received signal power to noise power ratio and is equal to $\left(\frac{\sigma^2}{P_r} + \frac{1}{\mathrm{SNR}_\mathrm{max} N_r}\right)^{-1}.$ Here, $P_r/\sigma^2$ is the actual signal to noise ratio (SNR) as defined next, and $\mathrm{SNR}_\mathrm{max} N_r$ limits the maximum possible received SNR with $N_r$ equal to the number of receiver antennas. A similar model for dampening very high SNR due to device imperfections is common in the industry, e.g. see the Qualcomm paper \cite{Zhang15}. It can be derived by modeling a virtual amplify-and-forward transmission hop within the receiving device, which leads to effective SNR being half of the harmonic mean of the actual and maximum SNR \cite[(4)]{Hasna03}. Note that for large $\mathrm{SNR}_\mathrm{max} N_r$, the effective SNR is close to $P_r/\sigma^2$. However, if $P_r/\sigma^2$ is itself very large, then the SNR cannot exceed $\mathrm{SNR}_\mathrm{max} N_r$. Note that $\mathrm{SE}_\mathrm{max}$ is the limit on maximum spectral efficiency, which is related to modulation and coding employed by the receiver. 
Here,  $P_r = \left(\frac{\lambda}{4\pi}\right)^2 \Upsilon \power{}\mathrm{N}_\mathrm{BS} \mathrm{N}_\mathrm{UE} x^{-\alpha}$, where $\power{}$ is the transmit power, $\sigma^2$ is the noise power, $\mathrm{W}$ is the bandwidth, $\mathrm{N}_\mathrm{BS}$ and $\mathrm{N}_\mathrm{UE}$ are the number of antennas at the BS and UE, $\lambda$ is the wavelength in meters, $\Upsilon$ is the blockage dependent correction factor\cite{KulVisAnd17}, and $\alpha$ is the blockage dependent path loss exponent (PLE). If the link is LOS, then $\alpha = \alpha_l$ and $\Upsilon = 1$. If the link is NLOS, then $\alpha = \alpha_n$ and $\Upsilon = \Upsilon_n\ll 1$. 

{\bf Remark 2} (Computing backhaul rate).
$R_1 = \mathrm{W}\min\left(\log_2\left(1+\mathrm{SNR}_b\right),\mathrm{SE}_\mathrm{max}\right)$, where $\mathrm{SNR}_b$ is half of the harmonic mean of $\frac{\left(\lambda/4\pi\right)^2 \power{}\mathrm{N}^2_\mathrm{BS} \mathrm{D}^{-\alpha_l}}{\sigma^2}$ and $\mathrm{SNR}_\mathrm{max} \mathrm{N}_\mathrm{BS}$. 
\begin{table}
\caption{Default numerical parameters}
\centering
\label{tab:param}
\begin{tabulary}{\columnwidth}{|C | C| C|| C| C| C|}
\hline
{\bf Notation} & {\bf Parameter(s)} & {\bf Value(s) if applicable} & {\bf Notation} & {\bf Parameter(s)} & {\bf Value(s) if applicable} \\\hline
$f_c$ & Carrier frequency & $28$ GHz\cite{PiKha} & $\mathrm{W}$ & Total bandwidth & $800$ MHz\cite{PiKha}\\\hline
$P_{d}$ & BS transmit power & $30$ dBm\cite{PiKha} & $P_u$ & UE transmit power & $23$ dBm\cite{PiKha} \\\hline
$\eta$ & Fraction DL UEs & 1 & $\sigma^2$& Noise power & $-174 + 10\log_{10}(\mathrm{W}) + 10$ dBm \\\hline
$\alpha_l$ & LOS PLE & 2\cite{Akd14} & $\alpha_n$ & NLOS PLE & 3.4\cite{Akd14}\\\hline
$\mathrm{N}_\mathrm{BS}$ & BS antennas & 64\cite{Nokia16} & $\mathrm{N}_\mathrm{UE}$ & UE antennas & 16\cite{Nokia16}\\\hline
$\mathrm{D}$ & ISD & $200$m & $k$ & Number of rings & 3 \\\hline
$\Upsilon_n$ & Correction factor & $-5$dB\cite{Du17,KulVisAnd17} & $\mathrm{SE}_\mathrm{max}$ & maximum spectral efficiency & 10 bps/Hz\cite{Eric17}\\\hline
\end{tabulary}
\end{table}

\section{Numerical Results and Design Guidelines Based on Analysis}
\label{sec:results}
In this section, we evaluate the derived formulae to explore system design insights for multi-hop mmWave cellular networks. In the next section, the main analytical assumption -- noise-limitedness -- will be validated. 
Table~\ref{tab:param} summarizes key parameters which are fixed throughout the numerical study, unless specified otherwise. NNHR is assumed, unless specified otherwise. We choose $\mathrm{SNR}_\mathrm{max} = 16$ dB, so that the maximum received SNR at UEs equals $28$ dB considering 16 antennas, which is close to the $30$dB value in \cite{Zhang15}. For backhaul links, the maximum received SNR is $34$ dB considering 64 antennas. For 5G-NR, it is possible to support up to 1024 QAM\cite{Eric17} and thus $\mathrm{SE}_\mathrm{max} = 10$ bps/Hz is chosen.

{\bf Fall in throughput with number of rings.}
To understand the fall in throughput with number of rings, we consider 2 worst case UEs per BS located at a distance $\mathrm{D}/2$ on the streets. LOS access and all DL UEs is assumed. Fig.~\subref*{fig:fallthroughput1} shows the fall in throughput with number of rings. It is surprising to note that it is possible to achieve minimum 100 Mbps per UE with even a 4 ring deployment, which covers an area of $1.1\times 1.1$ km$^2$ and supports 40 relays per fiber site. Having a larger $\mathrm{N}_\mathrm{BS}$ hardly changes the rate as the network is backhaul limited with backhaul links operating at $\mathrm{SE}_\mathrm{max}$. 
\begin{figure*}
  \centering
\subfloat[2 LOS UEs per BS located at 100 m from the serving BS.]
{\label{fig:fallthroughput1}{\includegraphics[width= 0.5\columnwidth]{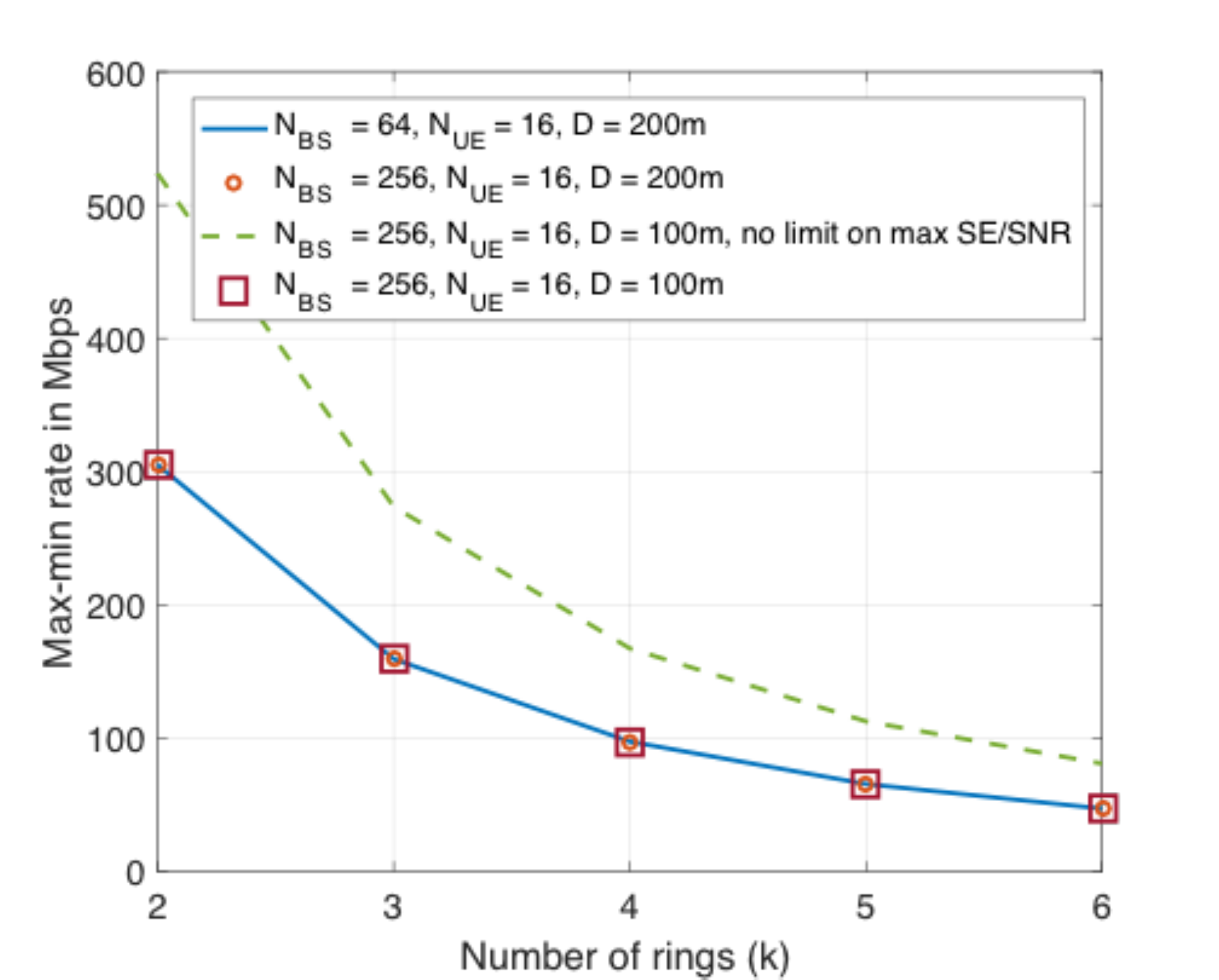}}}
\subfloat[2 NLOS UEs per BS located at 100 m from the serving BS.]
{\label{fig:fallthroughput2}{\includegraphics[width=0.5\columnwidth]{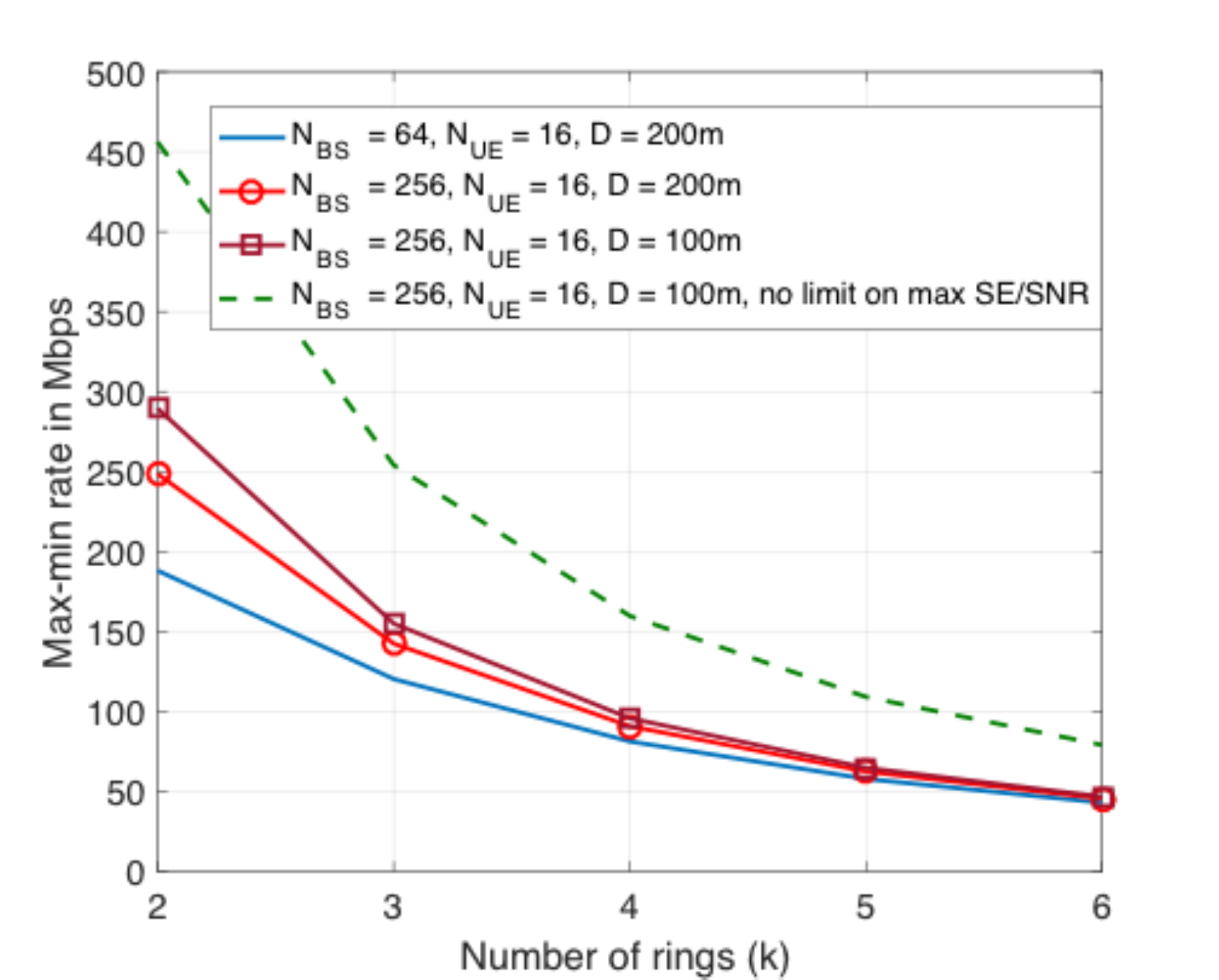}}}
\caption{Fall in throughput with $k$.}
 \label{fig:Fig12}
\end{figure*}
Decreasing $\mathrm{D}$ to 100 meters also does not change the rates. As per Corollary~\ref{cor:simpleformula}, throughput decays as $\frac{1}{w}\left(\frac{1}{R_a} + \frac{2k(k+1)}{R_1}\right)^{-1}$. Since we consider LOS UEs, $\mathrm{R}_a$ is already saturated by $\mathrm{SNR}_\mathrm{max}$ for $\mathrm{D}=200$m. Also, $R_b$ is limited by $\mathrm{SE}_\mathrm{max}$ and does not change by decreasing $\mathrm{D}$. However, note that 2 UEs per BS with $\mathrm{D} = 100$m itself supports $4$x higher user density than for $\mathrm{D} = 200$m. If there were no limit on spectral efficiency or SNR, then even up to $k=6$ with $\mathrm{D} =  100$m, that covers an area of $850\times850$ m$^2$, can be supported with user density of 200 UEs/km$^2$. This result motivates supporting higher order modulations and enabling the use MIMO on backhaul links to increase the spectral efficiency for enabling the support of higher load per BS or larger value of $k$ for the same target per user rate. 

As can be seen in Fig.~\subref*{fig:fallthroughput1}, throughput decays quickly with $k$ as the networks are backhaul limited. For large values of $k$, when the $1/R_a$ term is negligible, throughput decays by a factor of $k/(k+2)$ as $k$ increments by 1. The $1/R_a$ factor makes throughput decay slightly slower than above for smaller values of $k$. More specifically, if one fits function $\alpha/k^\beta$ to the plot for $\mathrm{N}_\mathrm{BS} = 64$ and $\mathrm{N}_\mathrm{UE} = 16$, then $\beta = 1.6$. The decay is slower in access limited networks, when $1/R_a$ term is non-negligible. This can be observed from Fig.~\subref*{fig:fallthroughput2}, which reproduces the scenarios in Fig.~\subref*{fig:fallthroughput1} but with NLOS UEs. Note that up to $3$ rings can be supported even with NLOS UEs. 

We now consider a more general UE deployment setup as shown in Fig.~\subref*{fig:softmaxmintopology2}. On average there are 2 UEs per BS in the 3-ring deployment. A random realization of LOS/NLOS states for UE to/from BS links was generated considering $50\%$ probability of being LOS within a distance of $200$m. Minimum path loss association is done. For the realization considered, $55\%$ UEs connected to LOS BSs. Also by default $\eta = 0.5$, that is about $50\%$ UEs are DL and rest are UL. Spectral efficiency (SE) has a minimum limit of $0.02$ bps/Hz below which rate is 0. 
\begin{figure*}
  \centering
  \subfloat[Topology under consideration.]
  {\label{fig:softmaxmintopology2}{\includegraphics[width= 0.5\columnwidth]{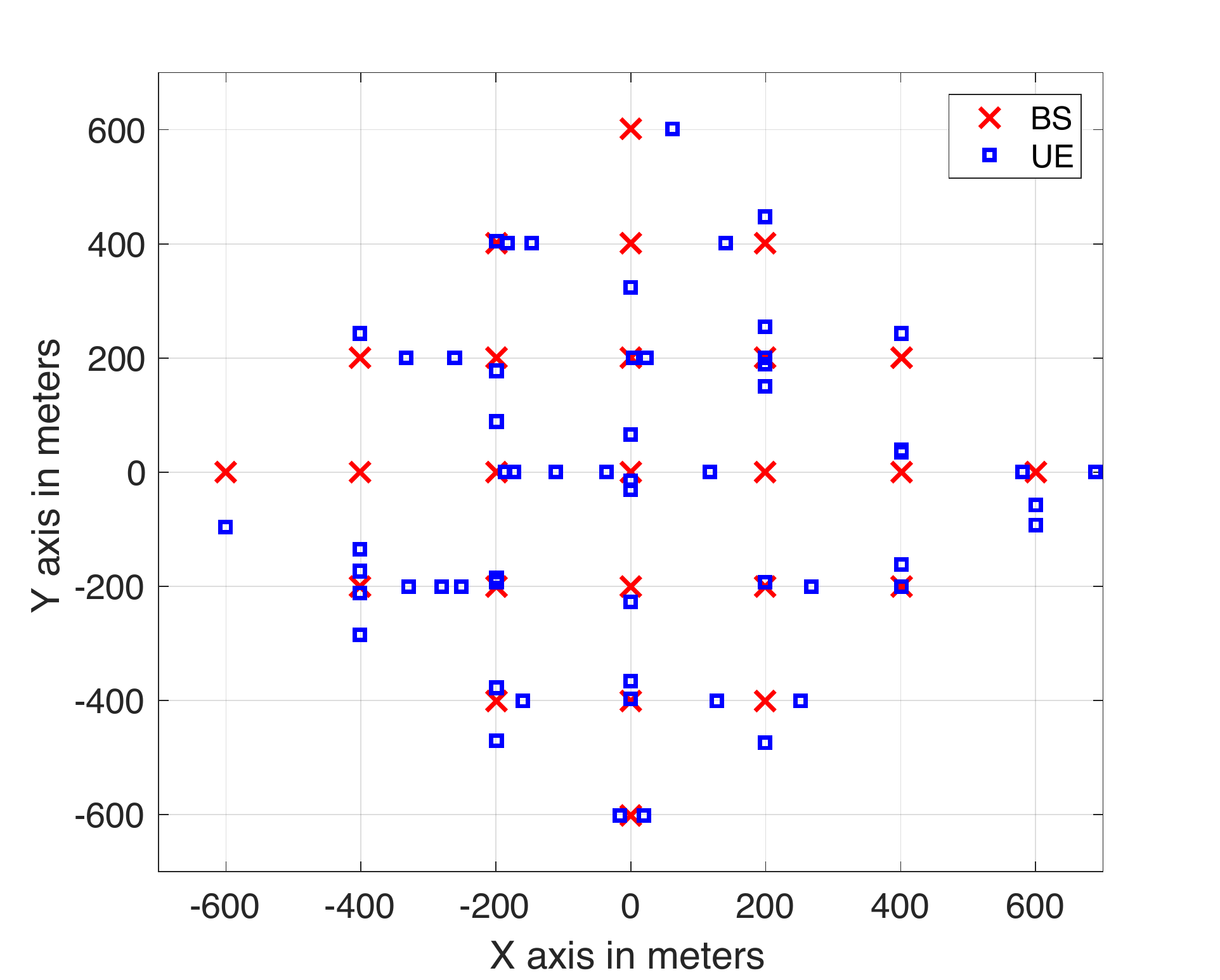}}}
\subfloat[Full versus half duplex relaying with soft max-min. ]
{\label{fig:Fig5}{\includegraphics[width=0.5\columnwidth]{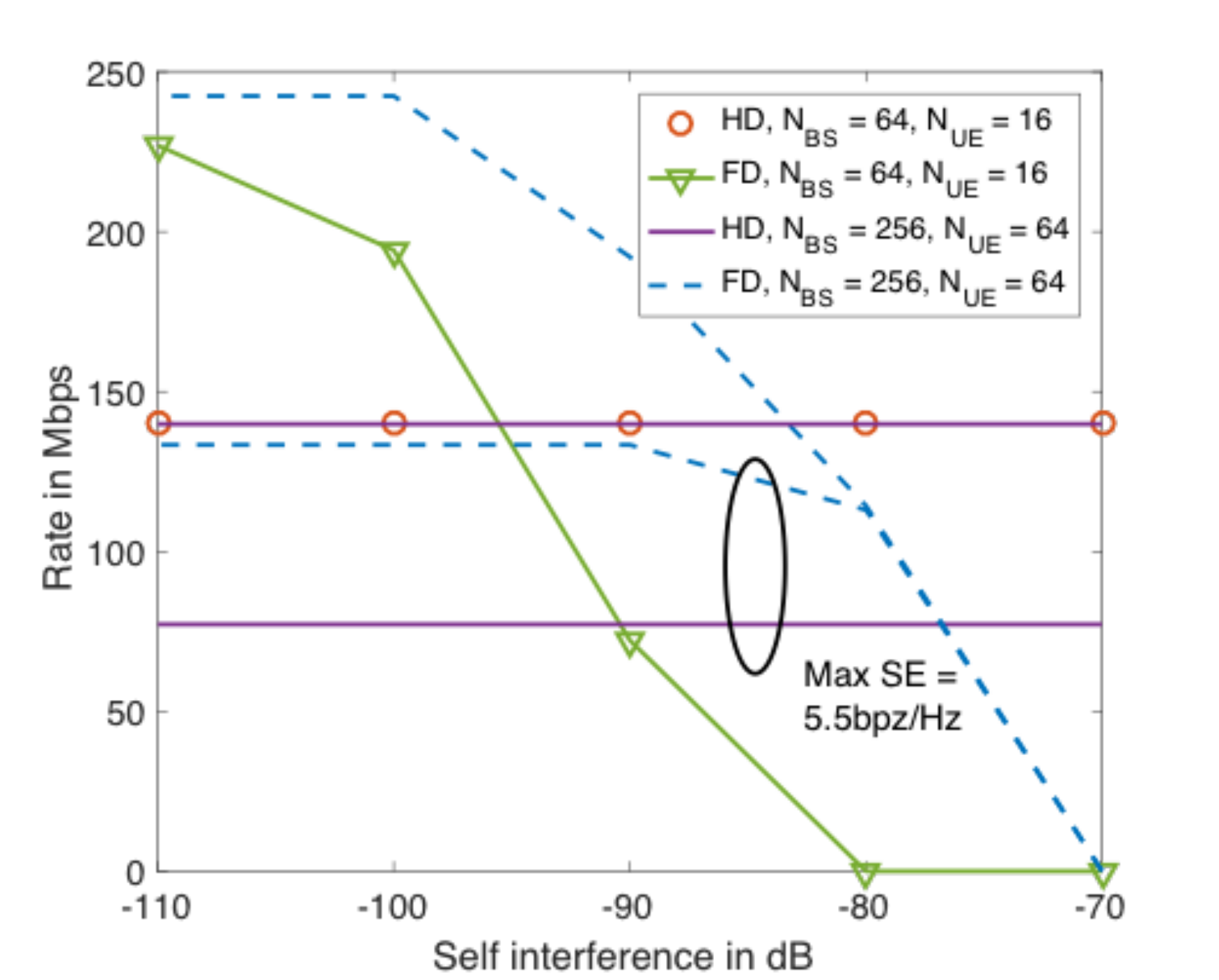}}}
 \label{fig:Fig45}
   \caption{Impact of full duplex relaying on max-min rates. Unless specified otherwise, the plots are generated assuming maximum SE is 10 bps/Hz. }
\end{figure*}

{\bf Impact of Full Duplex Relays.}
Fig.~\subref*{fig:Fig5} shows the comparison of full and half duplex relaying. X axis is the self-interference (SI) introduced by full duplexing and Y axis plots the optimal rates in Mbps. We consider soft max-min optimization, introduced in Section~\ref{sec:softmaxmin}, wherein $10\%$ of bad UEs are replaced with pseudo UEs that fake an arbitrarily large rate. We consider soft max-min since we observe that considering max-min optimization in the considered setup leads to a conclusion that full duplexing can provide higher rates than half duplex only if SI is less than $-110$dB, which is impractical to achieve as per state of the art prototypes\cite{Dinc17}. Fig.~\subref*{fig:Fig5} explores scenarios wherein larger SI can be tolerated. Even with soft max-min optimization, significant gains with full duplexing are observed for the default setup only if SI$<-100$dB. Fig.~\subref*{fig:Fig5} shows that considering larger antenna gains at the BSs and UEs helps increase the requirement of maximum tolerated SI to $-90$dB. Considering a maximum spectral efficiency of $5.5$ bps/Hz further increases the tolerance of SI to $-80$dB, which is practical\cite{Dinc17}. Note that $5.5$bps/Hz corresponds to spectral efficiency with 64 QAM and light coding. Similar values of $\mathrm{SE}_\mathrm{max}$ have been used in prior work\cite{Akd14,Mog07}. We next turn our attention to understanding if OAB can closely follow the rates obtained using IAB.

{\bf OAB versus IAB.} 
The distribution of e2e rates obtained using OAB is compared with IAB in Fig.~\subref*{fig:Fig8}. We consider two types of OAB. First allocates equal backhaul rate to each relay (called type 1). Second type offers a backhaul rate $w_{i,j}\gamma$ to a relay at $(i,j)$, wherein maximum $\gamma$ is computed (called type 2). The max-min rates with IAB outperform the rate obtained by more than $60\%$ of UEs with OAB type 1. Although not shown in the plot, varying $\zeta\in(0,1)$ does not change this insight. However, it is interesting to note that with OAB type 2 it is possible to achieve rates slightly greater than IAB rates for about $85\%$ UEs by choosing $\zeta = 0.15$. This is encouraging for practical implementations since OAB type 2 requires less global information for performing the optimization as compared to IAB. 
\begin{figure*}
  \centering
\subfloat[OAB type 2 closely follows IAB.]
{\label{fig:Fig8}{\includegraphics[width=0.5\columnwidth]{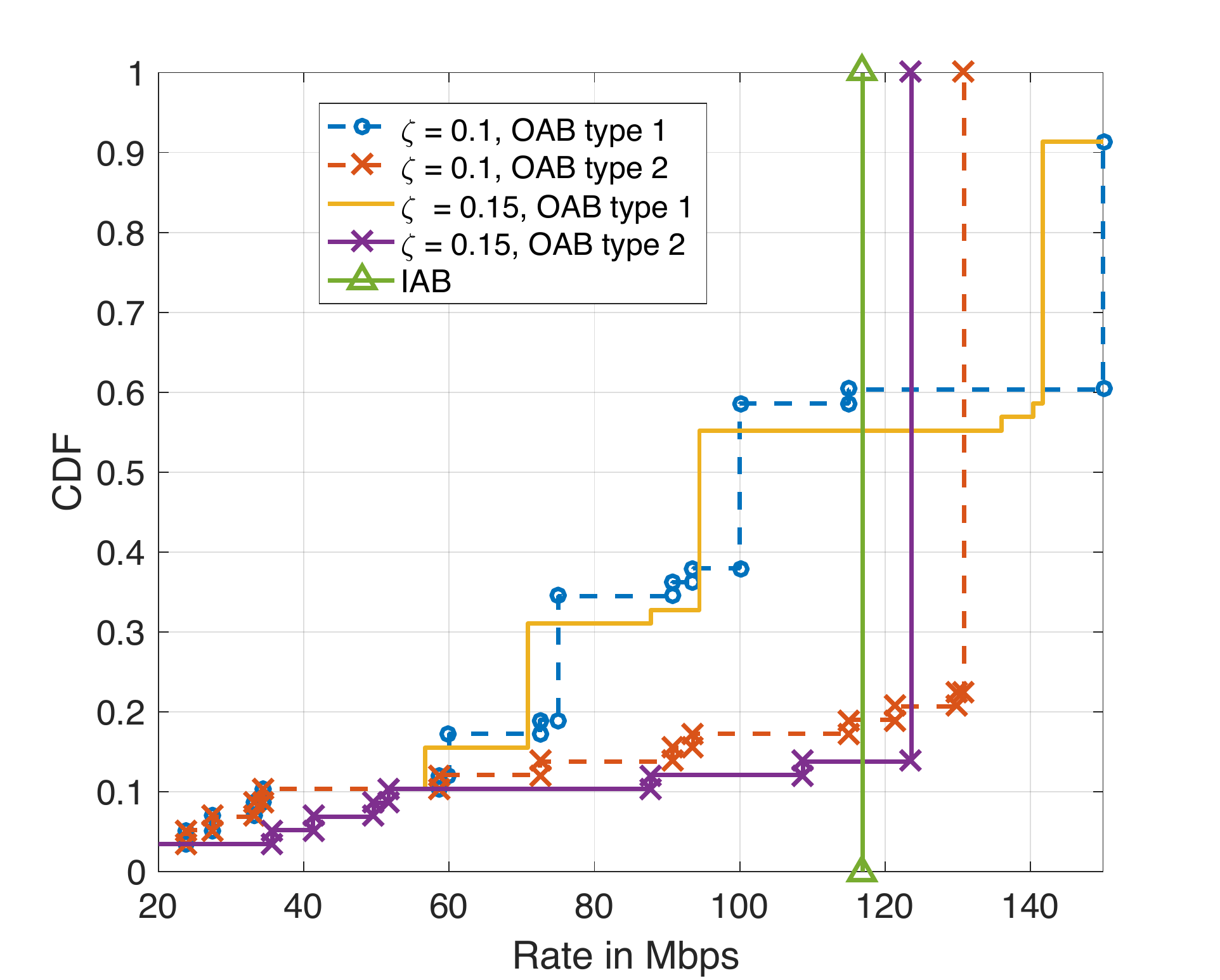}}}
\subfloat[Impact of dual connectivity on rate.]
{\label{fig:Figdual}{\includegraphics[width= 0.5\columnwidth]{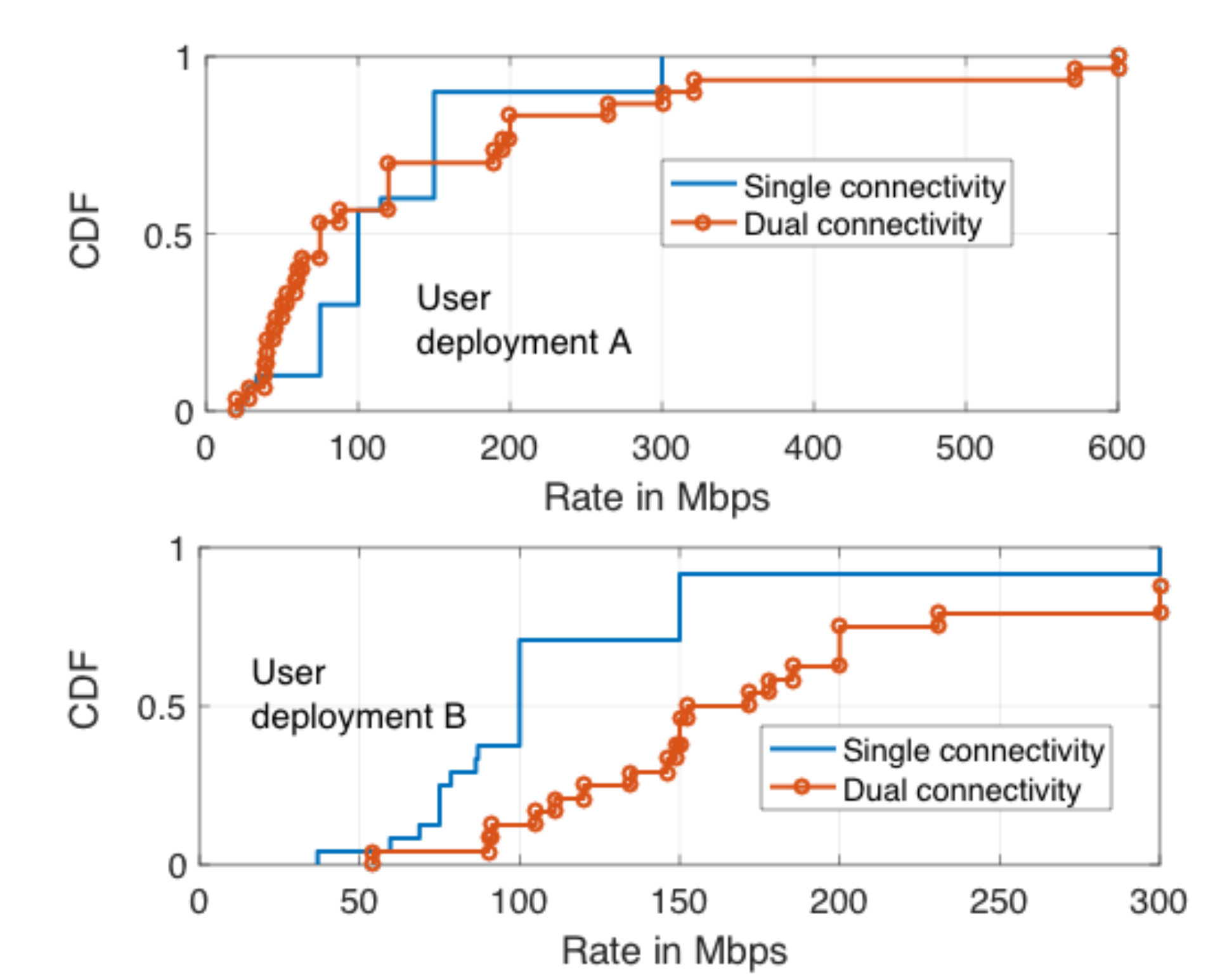}}}
\caption{OAB vs IAB, and impact of dual connectivity.}
 \label{fig:Fig9dual}
\end{figure*}

{\bf Dual connectivity versus single connectivity.}
We conclude our discussion of design insights based on the analysis by evaluating the benefit of DC as described in Section~\ref{sec:dualconn}. Fig.~\subref*{fig:Figdual} plots the rates with SC and DC considering two deployments and OAB type 1. Deployment A is the one in Fig.~\subref*{fig:softmaxmintopology2}, wherein there are about 2 UEs per BS on an average with a load variance of $1.1$. Deployment B is not shown for space constraints and has same mean UEs per BS but variance is $2.3$. For deployment B, median rates with DC are almost $1.5$x higher than SC. Although the load per BS is higher with DC, the load imbalance across BSs is reduced. Since equal backhaul rate per relay is offered, load balancing makes it possible to exploit the underutilized backhaul links. However, note that the rates with DC are roughly similar to SC for deployment A with lower load imbalance of UEs across BSs. We observe that in general the higher the load imbalance with SC, the higher the gain in data rates with DC. For deployment B, the max-min rate considering IAB is 133 Mbps and as can be seen OAB type 1 along with DC (both of which do not assume any knowledge of the network load) can enable $75\%$ of UEs to achieve this rate. 

\section{Validation of Noise-limitedness.}
\label{sec:validation}
Same default parameters as Table~\ref{tab:param} are used in this section. The goal is to motivate why noise-limited analysis works through a couple of empirical observations. Also, we observe NNHR operates optimally even in more general scenarios than in Theorem~\ref{thm:IAB1}. We also propose a greedy variant of proportional fair (PF) scheduling for multi-hop networks in one of the numerical examples that is used to validate noise-limited analysis. This example is also useful to show how the analysis can be used as a benchmarking tool for complex simulators.  All UEs are DL. 

\subsection{Few bottleneck links helps noise-limitedness.}
\label{sec:fewlinks}
We compare the max-min rate obtained from our noise-limited analysis with that computed using the linear programming (LP) solution in \cite{Ras15}, which jointly optimizes the scheduling and routing. An arbitrary deployment was considered and interference was not neglected in \cite{Ras15}. This, however, lead to a LP formulation with very high numerical complexity as compared to our work. Specifically, if there are $L$ links in the network one needs to create matrices of the size on the order of $2^L$ to implement the LP. 

{\bf Simulation setup.}
We consider the deployments in Fig.~\subref*{fig:softmaxmintopology2} and Fig.~\subref*{fig:validationtopology} with average loads equal to 2.3 UEs per BS and 2.6 UEs per BS. Inter-site distance is 200 m in  Fig.~\subref*{fig:softmaxmintopology2} and 100 m in Fig.~\subref*{fig:validationtopology}. Searching over all possible routes is not possible using the algorithm in \cite{Ras15} considering that the network in Fig.~\subref*{fig:softmaxmintopology2} has $25$ BSs and $58$ UEs. We reduce the search space by considering only NNR (not necessarily highway routing) on the grid. %Since links across orthogonal streets can have very high path loss\cite{Kart17} and long links on the same street would tend to be NLOS, this is likely not a bad assumption for mmWave mesh networks. 
Since listing all scheduling patterns given NNR is itself time and memory intensive (there are 86 valid links in Fig.~3(a) even after reducing the search space for routing and there will be on the order of $2^{86}$ potential schedules), we do a greedy search to list transmission schedules. We greedily list $900$ transmission schedules that have at least 10 active links such that every node has at most one active at a time. Since we expect the bottleneck node to be the fiber site, $450$ of these schedules have at least one backhaul link connected to the fiber site. Along with the greedy schedules, we also include all transmissions schedules which have exactly 3 active links, respecting the half duplex constraint of the users and base stations, in the search space to make sure the LP has a solution. The greedy schedules were considered to check if the optimal scheduler ignoring interference prefers these schedules over the rest of the schedules which have only 3 active links at a time.
\begin{figure*}
  \centering
\subfloat[Validation topology.]
{\label{fig:validationtopology}{\includegraphics[width=0.5\columnwidth]{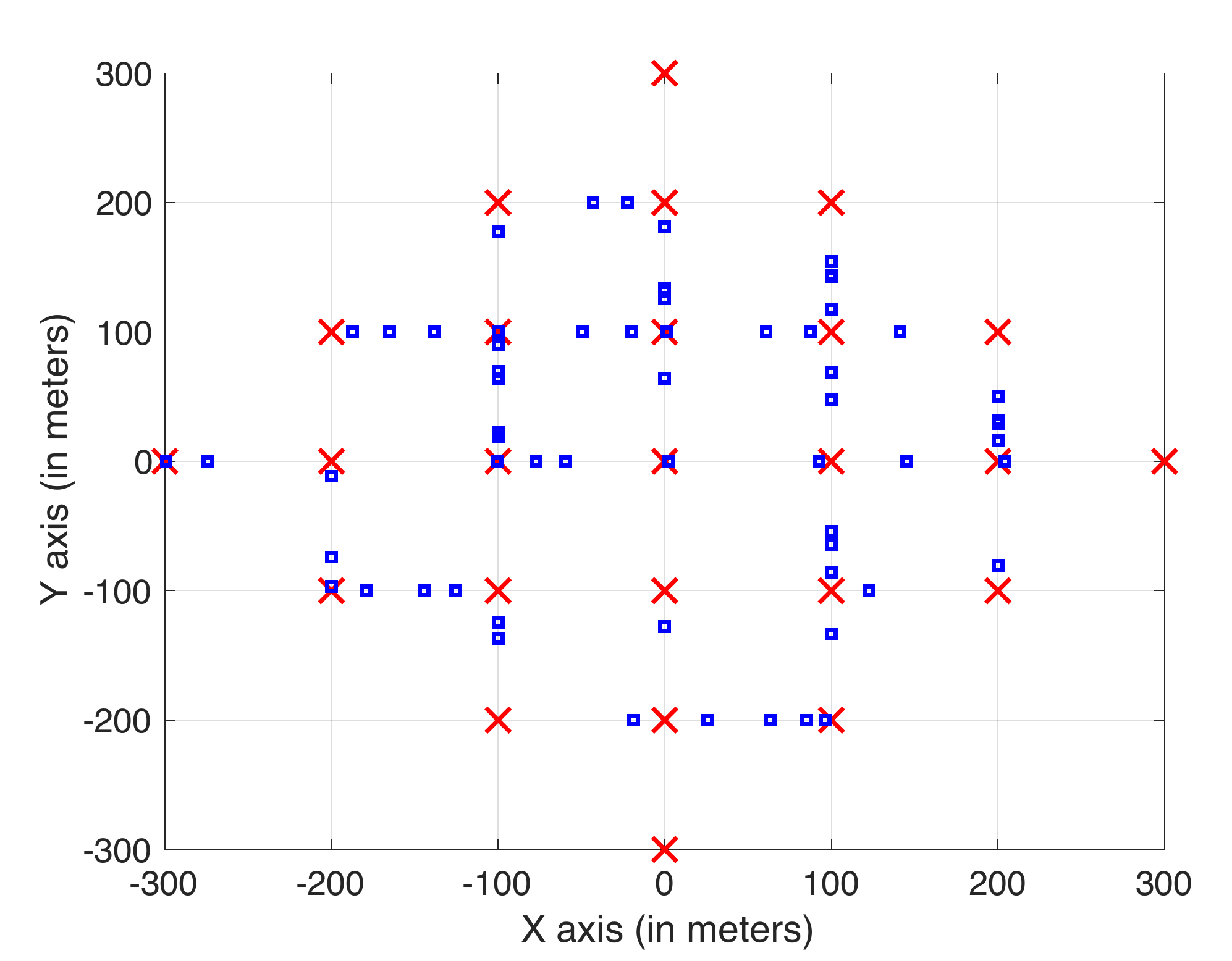}}}
\subfloat[Performance of proposed greedy PF scheduler. I and NI indicate that interference was included and no interference. ]
{\label{fig:Fig10}{\includegraphics[width= 0.5\columnwidth]{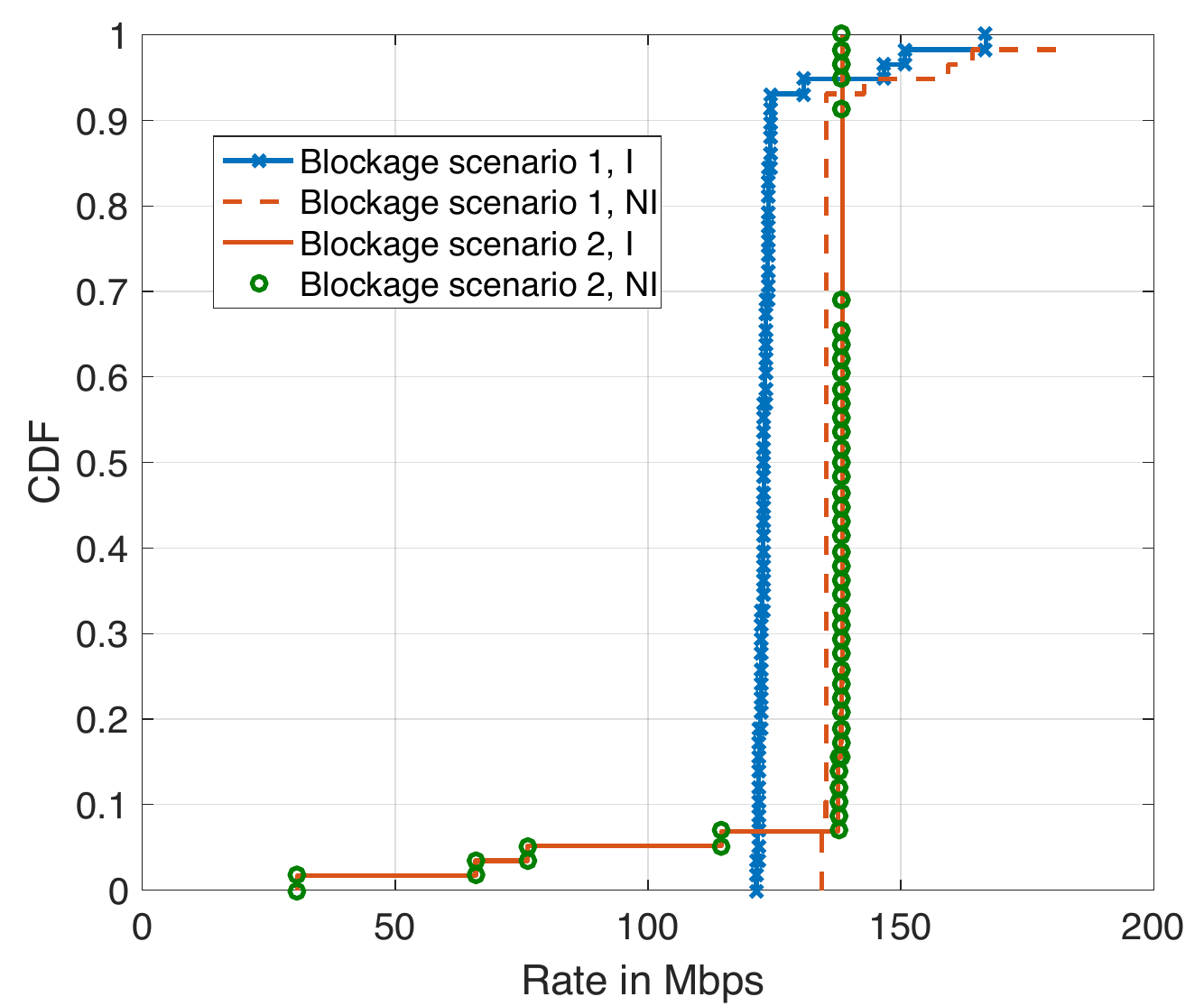}}}
\caption{Validation plots. }
 \label{fig:Fig1410}
\end{figure*}

To model interference, we consider received signal power from interfering transmitters as 
$P_r = \left(\lambda/4\pi\right)^2 \power{} G_t G_r x^{-\alpha},$
where $G_t$ and $G_r$ are maximum antenna gains of the transmitter and receiver, and rest of the parameters are defined in Remark~1.  All links along a street are assumed to be LOS within 200 meters distance. Note that as per the results in \cite{Akd14,SinJSAC14,5GWhitePaper2015}, this is a worst case assumption in terms of interference (not in terms of desired signal power). Here, $G_t, G_r\in\{\mathrm{N}_\mathrm{BS},\mathrm{N}_\mathrm{UE}\}$ depending on whether the transmitter or receiver is a BS or UE. For links across different streets, a NLOS path loss exponent of $3.4$ is used. Note that we are not exploiting narrowness of beams from all interferers and thus the interference we consider can be interpreted as worst case. Furthermore, note that we are not exploiting that NLOS path loss can be as high as $10$ for orthogonal street links\cite{Kart17}.

{\bf Comparing max-min rates with and without interference.}
Table~\ref{tab:validation} summarizes the results of max-min rates obtained by running the LP in \cite{Ras15} for different scenarios. Specifically, max-min rates were computed assuming noise-limitedness and considering interference. Also, max-min rates were computed assuming NNHR and without such an assumption. Here, $t_1$ denotes the fraction of time that the greedy schedules were used, so $1-t_1$ is the fraction of time 3 links were active at a time. The fraction of time wherein at least one link connected to the fiber site was active is denoted by $t_2$.

It is surprising at first to note that irrespective of whether interference is considered or not, the max-min rates do not change for scenario in Fig.~\subref*{fig:softmaxmintopology2}. For scenario in Fig.~\subref*{fig:validationtopology}, which is a much denser network, there is only about $10\%$ loss in max-min rate inspite of considering a worst case interference model. These observations are explained by noting the values of $t_1$. Since $t_1\ll1$, most of the schedules used only 3 active links at a time to meet the max-min rates. In other words, the optimal scheduler hardly used transmission schedules with greedy packing. Thus, the interference is negligible in these scenarios. As mentioned in \cite{Ras15} there is no unique solution to the LP and thus the values of $t_1$ are not unique.  The key takeaway, however, is that there exists a solution that achieves max-min rates under the scheduling and routing search space considered with small $t_1$.  Furthermore, $t_2$ close to $1$ implies the bottleneck node is the fiber site. This exercise highlights the importance of our noise-limited analysis as a tool to evaluate max-min rates in closed form, even though actual schedulers may need to be interference-aware. Furthermore, the rates considering NNHR closely follow those obtained considering an optimal NNR strategy in Table~\ref{tab:validation}. Since the load conditions in Fig.~\subref*{fig:softmaxmintopology2} and  Fig.~\subref*{fig:validationtopology} do not follow the strict load constraints in Theorem~\ref{thm:IAB1} for optimality of NNHR and also all access links do not have same rate unlike in Theorem~\ref{thm:IAB1}, we conclude that the constraints for optimality of NNHR in Theorem~\ref{thm:IAB1} may be further relaxed. 
\begin{table}
\caption{Empirical evidence for noise-limitedness considering worst case interference model.}
\centering
\label{tab:validation}
\begin{tabulary}{\columnwidth}{|C|C| C| C|| C| C| C|C|}
\hline
{\bf Scenario (Fig.~\subref*{fig:softmaxmintopology2})} & {\bf $t_1$} & {\bf $t_2$} & {\bf Max-min rate (Mbps)} & {\bf Scenario (Fig.~\subref*{fig:validationtopology})} & {\bf $t_1$} &{\bf $t_2$} & {\bf Max-min rate (Mbps)} \\\hline
Optimal NNR + interference & 0.08 & 1 & 136.77 & Optimal NNR + interference &     0.01
 & 0.95 & 119.70\\\hline
Optimal NNR without interference &  0.07 & 1 & 137.17 & Optimal NNR without interference &     0.01
& 0.99 & 128.31\\\hline
NNHR + interference & 0.08 & 1 & 136.65 & NNHR + interference & 0.1 & 0.91 & 114.74\\\hline
NNHR without interference & 0.03 & 1 & 137.17 & NNHR without interference &    0.01 & 1 & 128.30\\\hline
\end{tabulary}
\end{table}

Similar observation highlighting noise-limitedness due to interference aware schedulers were reported in \cite{Ras15,Rois15}. If one comes across a deployment and traffic scenario wherein the rates with NNHR are much lower than optimal NNR, then the methodology mentioned in Section~\ref{sec:IAB} to increase the antenna gains on highway relays can be attempted. Our code for implementing the LP in \cite{Ras15} is available at \cite{KulCode18}. 

{\bf Remark 3.} 
Since our analytical results with NNHR give the same rate as that obtained by employing the LP in \cite{Ras15}, the results for max-min rate in Table~\ref{tab:validation} are accurate in spite of a small search space. Furthermore, increasing the number of greedy schedules to 1800 did not change the result for the case of NNR without interference, making us confident  that the result is not affected by the choice of small search space. 

Noise-limitedness is observed in spite of  worst case interference assumption since it turns out that activating only 3 links at a time is sufficient for the optimal scheduler given nearest neighbour routing. Interference in the network is small since the scheduler picks simultaneously active links which are far enough and thus the scheduler is {\em interference-aware}. In other words, although the scheduler design does not ignore interference effects, the network effectively appears to be noise-limited in terms of max-min rates. We now report another interesting observation, wherein we utilize the directionality and blockage effects to show that DL $k-$ring networks can be noise-limitedness even if scheduler design does not explicitly take interference effects into account.

\subsection{Blockage effects and directionality helps noise-limitedness.}
\label{sec:PF}
In this section, we show that the blockage effects at mmWave along with directionality of transmissions in the $k-$ring deployment can enable noise-limited analysis even if the scheduler does not explicitly protect interference on bottleneck links. We now assume that transmitters on different streets than the receiver have negligible interference, since the path loss exponent can be as large as $10$ for the NLOS segments of such links\cite{Wang18,Kart17}. To model interference, we consider received signal power from interfering transmitters as 
$P_r = \left(\lambda/4\pi\right)^2 \power{} G_t G_r x^{-\alpha},$
where $G_t$ and $G_r$ are random antenna gains and rest of the parameters are defined in Remark~1. All links along same street are LOS and all links across different streets are NLOS. Here, $G_t = G_\mathrm{max}$ if the interfering link is pointed exactly towards the receiver and $G_t = G_\mathrm{min}$, otherwise. Similarly, $G_r = G_\mathrm{max}$ is the receiver under consideration has beam pointed towards the interferer and $G_r = G_\mathrm{min}$ otherwise. Note that we have only 4 directions to point in the grid deployment. Here, $G_\mathrm{max}\in\{\mathrm{N}_\mathrm{BS},\mathrm{N}_\mathrm{UE}\}$ depending on whether the transmitter or receiver is a BS or UE, and $G_\mathrm{min} (\text{dB}) = G_\mathrm{max}(\text{dB}) - 25$dB, where $25$ dB is the front to back ratio. Such model was used previously in several works, see \cite{AndBai16} for relevant literature survey. 

We simulate the performance of the deployment in Fig.~\subref*{fig:softmaxmintopology2} using NNHR and a greedy variant of the popular backpressure scheduler with congestion control on the first hop as in \cite{Akyol08}. We call this as the greedy PF scheduler and it greatly simplifies the implementation of the GBD algorithm in Section I-C of \cite{Akyol08}. We choose this particular baseline algorithm, since it emulates PF for multi-hop networks with the utility function in Section I-C of \cite{Akyol08} being $U(x) = \log(x)$, which has been a popular paradigm for employing in 4G cellular networks. Another reason for choosing this scheduler is that the discussion in this work is limited to full buffer assumption until now, and considering a scheduler that works for time varying traffic is desirable. This would pave a way for evaluating packet latencies in multi-hop mmWave networks. However, in this section we assume the fiber site always has infinite backlogged data for all UEs. Each BS in the k-ring deployment now represents a queue with multiple traffic flows, each UE representing a flow. Here, we simulate the queueing network for a reasonably long time to understand whether directionality and blockage effects helps keep the network noise-limited even with the proposed simplified scheduler which is not interference-aware.
\begin{figure*}
  \centering
\subfloat[Per UE Access SINR/SNR.]
{\label{fig:Fig11}{\includegraphics[width=0.5\columnwidth]{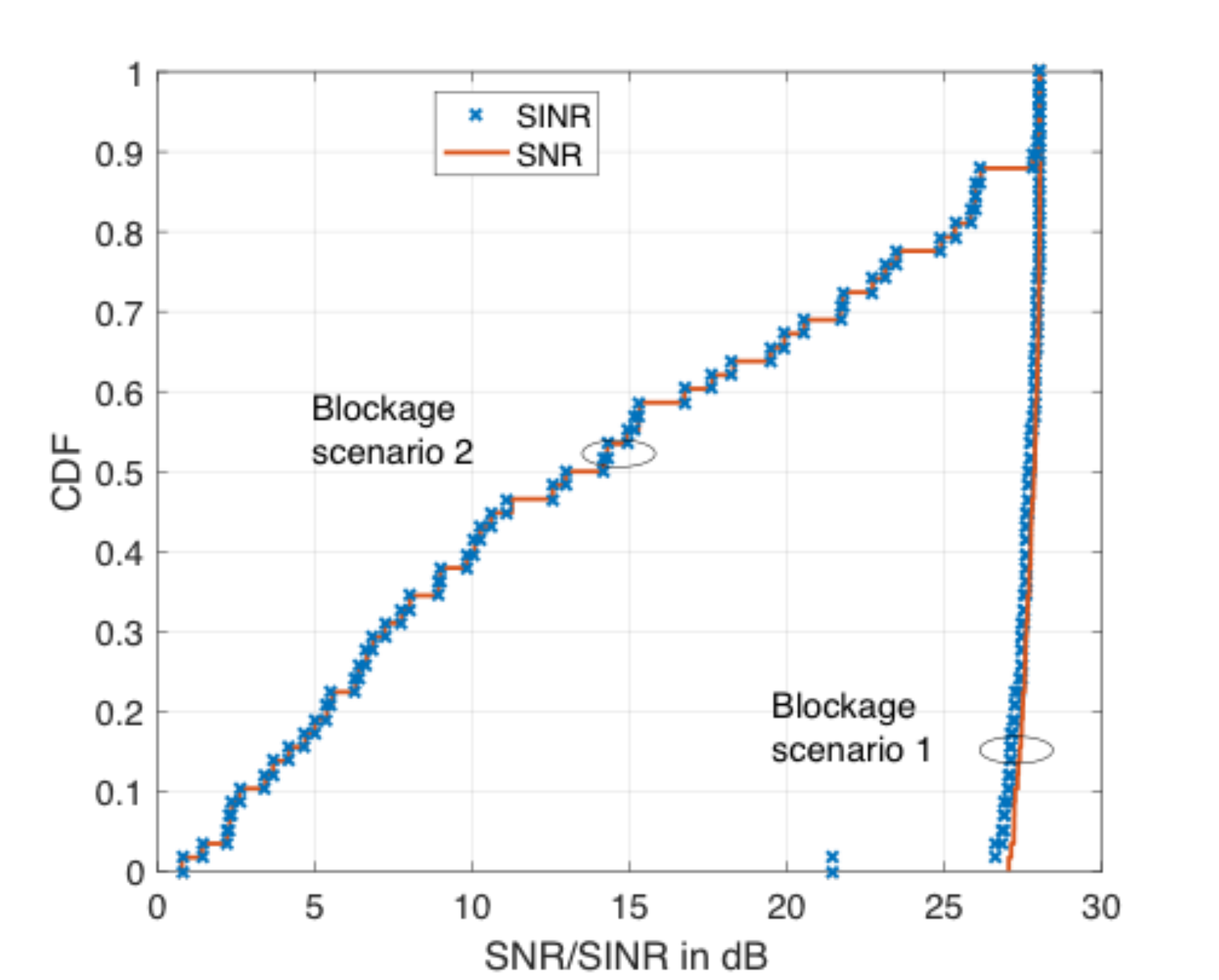}}}
\subfloat[Per UE Backhaul SINR/SNR.]
{\label{fig:Fig13}{\includegraphics[width= 0.5\columnwidth]{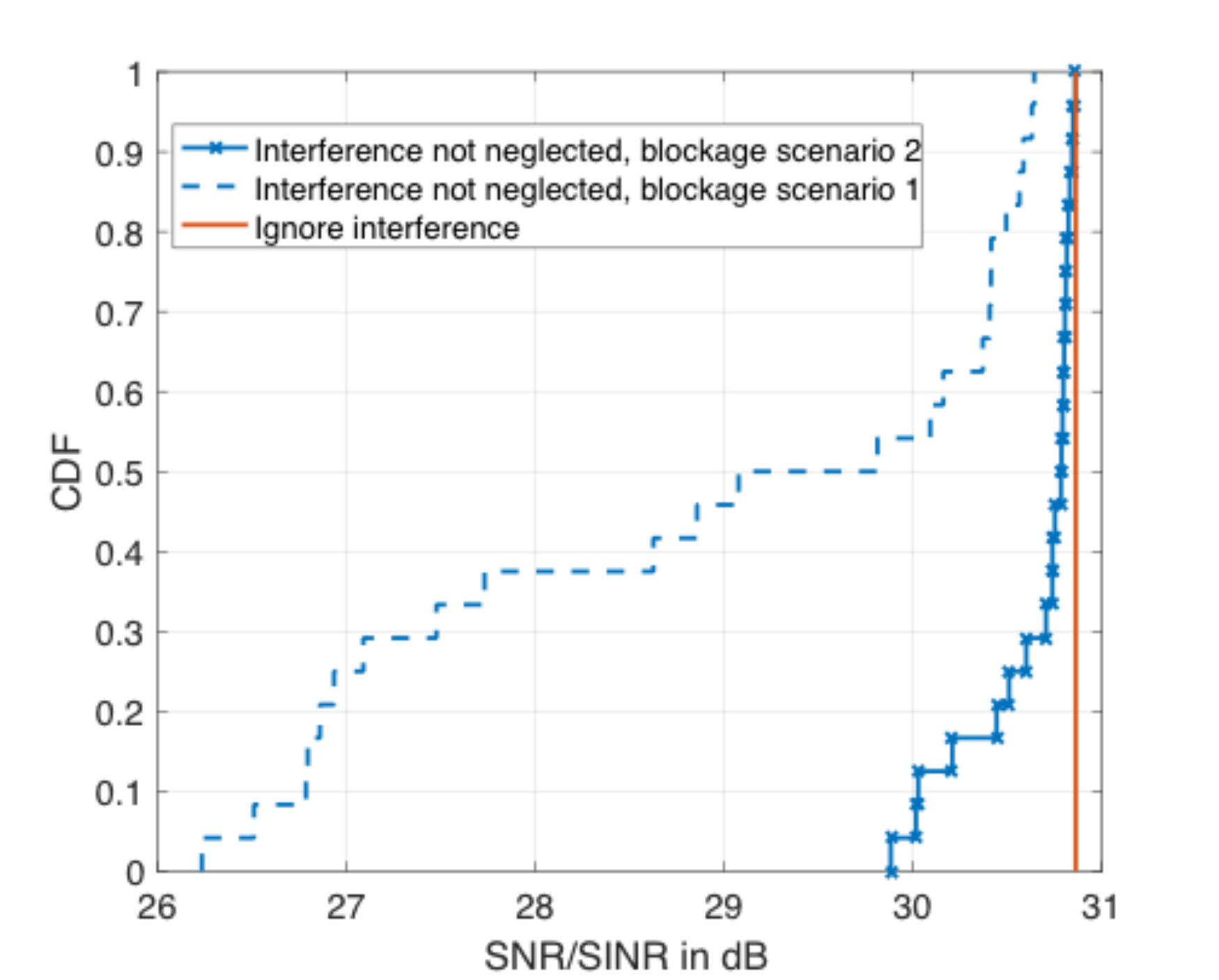}}}
\caption{SINR vs SNR considering the greedy PF scheduler. }
 \label{fig:Fig1113}
\end{figure*}

Assuming NNHR, scheduling is done as follows. We assign priority of scheduling {\em a particular flow on a particular link} by the {\em backpressure} metric, that is product of noise-limited estimate of the instantaneous rate on the link times difference in queue length at the source and destination for the flow on that link. For full buffer traffic, the queue length at source is infinite, which is why the backpressure metric needs modification on the first hop. Congestion control is done as discussed in \cite{Akyol08} to emulate PF scheduling and the priority metric for flow $f$ on link $i$, which corresponds to first hop for flow $f$, is $r_i(t)\left(\frac{1}{cR^{f}_i(t)} - q^{f}_i(t)\right)$, where $r_i(t)$ is the noise-limited instantaneous rate of link $i$ at time $t$, $q^f_i$ is the queue length for flow $f$ at the destination of link $i$ and $c$ is the congestion control parameter (set to be $10^{-14}$ to create high backpressure on the first hop for all UEs). Here, $R^{f}_i(t) = \beta R^{f}_i(t-1) + (1-\beta)\delta^{f}_i(t-1) r^{f}_{i,\text{actual}}(t-1)$, with $\beta = 0.99$ and $\delta^{f}_i(t-1)$ being the indicator that link $i$ was scheduled for flow $f$ in time slot $t-1$. Here, $r^{f}_{i,\text{actual}}(t-1)$ is the actual data rate of the scheduled flow $f$ on link $i$ at time $t-1$ (considering interference that resulted as an outcome of the scheduling decision in the previous time slot). Under the assumption of NNHR, scheduling is done using the computed priorities as follows. We pack in links with at least one non-zero priority flow in descending order of the highest priority flow through a link, respecting the half duplex constraint of the devices. If a link is scheduled as per this criterion, then the flows corresponding to highest priority on the scheduled links are chosen in that particular slot for scheduling. This is a greedy variant of the algorithm considered in \cite{Akyol08} since instead of searching over all possible transmission schedules we pick a greedy schedule in descending order of priorities. 

Fig.~\subref*{fig:Fig10} shows the distribution of achieved per user e2e rates for the topology in Fig.~\subref*{fig:softmaxmintopology2}, which is computed by dividing the total number of bits received by the UEs from the MBS during a simulation run of 10000 iterations with $0.2$ ms slot duration. Data for SINR/SNR and e2e rate was collected after 1000 warm-up iterations and queues at the relays were empty initially. Two blockage scenarios were considered. Scenario 1 implies all links along the same street are LOS to create a worst case interference scenario. Scenario 2 implies only neighbouring backhaul links are LOS, and rest are NLOS. This is reasonable since the non-neighbouring BSs are at least 400m apart, which will likely lead to NLOS links\cite{Akd14,Kart17}. Scenario 2 also assumes that access links are all NLOS to generate a scenario with low access SNRs. Fig.~\ref{fig:Fig1113} shows the access and backhaul SINR versus SNR comparison in the two blockage scenarios. It can be seen from Fig.~\subref*{fig:Fig10} that the rate distribution is almost vertical, implying equal rate per UE was achieved. Ignoring interference, the rate is exactly equal to the max-min rate from our analysis (also equal to $137$ Mbps as in Table~\ref{tab:validation}), which is surprising at first but can be explained as follows. The bottleneck links for all UEs are those connected to the fiber site having a constant rate $R_1$. Thus, irrespective of whether we do PF or max-min fair scheduling the rates coincide\footnote{Since UE rates on the fiber backhaul links add up to a constant, the solution of PF is same as max-min.}. Note that there is a small drop in rates (by $10\%$) with interference when all interferers on the same street are LOS. In blockage scenario 2, it is found that the rates with interference do not change at all. This confirms the noise-limited behaviour of the network under consideration. %Similar observations were made considering scenarios with some fraction of access links being LOS and rest being NLOS. 

An intuitive observation that explains noise-limitedness for the DL scenario is as follows. All backhaul links transmit in a direction away from the fiber site. This along with NNHR ensures that the bottleneck backhaul receivers connected to the fiber site never see interference from beams aligned towards their receivers from other backhaul transmitters. Nearby access links pointing towards ring 1 are also rarely activated since the ring 2 relays have to backhaul traffic for ring 3 relays. This naturally leads to low interference on bottleneck links.
 
%We note that a maximum limit on SE and SNR inherently provides interference protection for very high values of SINR, even if the actual interference is not negligible. With dense deployment and large antenna arrays (say, $\mathrm{D}\leq 200$m, $\mathrm{N}_\mathrm{BS}\geq 64$ and $\mathrm{N}_\mathrm{UE}\geq 16$), this phenomenon may occur routinely, especially in LOS environments. However, we have checked that even if $\mathrm{SE}_\mathrm{max}$ and $\mathrm{SNR}_\mathrm{max}$ are set to be impractically high, the insights in Section~\ref{sec:validation} do not change. An exception would be the noise-limitedness insight for blockage scenario 1 may not hold with greedy PF scheduler. This makes accurate modeling of the blockages and directionality crucial to get back noise-limitedness. 

\section{Conclusions and Future Work}
\label{sec:conclfuture}
A baseline model to study e2e rates in self-backhauled mmWave networks with multiple nearest neighbour backhaul hops is presented. The model is used to derive maximum achievable rates by all UEs in the network as a function of number of relays per fiber site. %We observe that it can be feasible to support up to 3 or 4 backhaul hops in self-backhauled mmWave networks with ISD between 100 to 200 meters, and still meet 100 Mbps e2e per user rates with about 2-4 full buffer UEs per BS. To increase the supported UE load on the $k-$ring network, advanced MIMO techniques or sectorization will be needed to increase the SE, especially on backhaul links. %However, noise-limitedness needs to reconfirmed in these scenarios. 
%Performance of full and half duplex relaying was compared and it was found that high antenna gains at BSs and UEs (say, $\mathrm{N}_\mathrm{BS}=256$, $\mathrm{N}_\mathrm{UE}= 16$) are required to tolerate self-interference on the order of $-80$dB and still offer $1.5-2$x gain in rates as compared to half duplex relays. 
Apart from the simplicity in the derived formulae, a key takeaway is that noise-limitedness in the $k-$ring deployment is aided by the observation that there are a few bottleneck links in the network making it sufficient for a max-min scheduler to activate only a few links at a time. %We believe that to bridge the gap in theory and practise while studying the performance of multi-hop mmWave networks it is essential to considering network planning along with resource allocation in order to exploit the unique interactions of mmWave propagation with the deployment scenarios considered. 

Understanding when would non-nearest neighbour strategies be desirable in mmWave networks is a scope for future work. Studying the robustness of access and backhaul link-failures in multi-hop mmWave networks is a possible scope of future research. If a backhaul link fails, there are two possible options -- re-route the traffic on an alternate path or use the NLOS backhaul link. If an access link fails, multi-connectivity with mmWave or sub-6GHz BSs can be explored to provide robustness to dynamic blockages. A practical challenge in the multi-connectivity scheme discussed in this work is that the scheme assumes that data can be received by a UE over two different streams without cross interference. More research needs to be done to understand the practical feasibility of such an assumption. Another direction would be to understand latency performance considering the $k-$ring deployment model. It is likely that max-min or PF schedulers have poor delay performance, and thus newer optimization problems that offer different rates to relays on different rings may be required to bound the delay performance within the 1ms latency requirement for 5G. 

\appendices
\section{Proof of Theorem~\ref{thm:IAB1}.}
\label{app:thmproof}
This proof has three parts. First we assume NNHR and derive an upper bound on e2e rates using \eqref{optformulation}. Then we show feasibility of a link activation strategy to achieve this upper bound using Algorithm~\ref{algo1}. Lastly, we show that no other routing strategy can perform better than NNHR under the assumptions in Theorem~\ref{thm:IAB1}.

Let $\gamma$ be the minimum e2e rate that all UEs can achieve. Knowing the instantaneous rates of the links and the loads, let us now write down necessary conditions for $\gamma$ to be minimum achievable e2e rate, assuming NNHR. The following inequality needs to hold considering the scheduling done by the MBS.
\begin{equation}
\label{eq:inequMBS}
\gamma\left(\frac{w_{0,0}}{R_a} + \frac{f(0,0)-w_{0,0}}{R_1}\right)\leq 1.
\end{equation}
Here, left hand side is the total time a BS is active (either transmitting or receiving) and right hand side is the total available time. Here, $\frac{\gamma}{R_a}$ is the minimum fraction of time utilized by MBS for serving a UE directly connected to it on access link. Since there are $w_{0,0}$ such UEs, $\frac{\gamma w_{0,0}}{R_a}$ is the minimum fraction of time MBS spends on access links. Similarly, $\frac{\gamma}{R_1}$ is the minimum fraction of time the MBS spends to serve any indirectly connected user by wireless backhauling. Since there are $f(0,0)-w_{0,0}$ such users, we have the required inequality.

Similarly, one can write down inequalities considering minimum fraction of time other BSs need to be active to allow $\gamma$ as the minimum achievable rate to all UEs. Considering the BS at $(i,j)$, with at least $i$ or $j$ not equal to $0$, the following inequality can be written. 
\begin{align*}
\gamma\left(\frac{w_{i,j}}{R_a} + \frac{f(i,j) - w_{i,j}}{R_1} + \frac{f(i,j)}{R_1}\right)\leq 1,
\end{align*}
where $ \frac{f(i,j) - w_{i,j}}{R_1}$ is the minimum fraction of time the BS has to allocate for backhauling to relays connected to it, further away from $(0,0)$, and $\frac{f(i,j)}{R_1}$ is the minimum fraction of time the BS is served by its parent BS towards the MBS. Since $w_{i,j} = w_{-i,-j}$ and NNHR, we have $f(i,j) =  f(-i,-j)$. Thus, the inequality can be written down as 
\begin{align}
\label{eq:inequrelay}
\gamma\left(\frac{w_{i,j}}{R_a} + \frac{f(-i,-j) + f(i,j) - w_{i,j}}{R_1}\right)\leq 1.
\end{align}
Since $w_{0,0}>w_{i,j}$ and $f(i,j) + f(-i,-j) - w_{i,j} \leq f(i,j) + f(-i,-j) \leq f(0,0) - w_{0,0}$, the inequality \eqref{eq:inequMBS} is stricter than \eqref{eq:inequrelay}. Thus, the bottleneck inequality is always \eqref{eq:inequMBS} and thus, $\gamma \leq \left(\frac{w_{0,0}}{R_a} + \frac{f(0,0) - w_{0,0}}{R_1}\right)^{-1}$.

If we prove that a scheduler with NNHR helps achieve the above upper bound, then $\gamma^*$ is the max-min rate.  Consider the scheduler in Algorithm~\ref{algo1}. If a UE is connected to the MBS, it is clear from the algorithm that its long term rate is $\gamma^*$ since the user gets $\gamma^*/R_a$ fraction of time with instantaneous rate $R_a$. If a UE is connected to the BS at $(i',j')$, then to ensure its long term rate is $\gamma^*$ we need all the backhaul hops to support at least $\gamma^*$ long term rate for the data of this particular user. Also we need the long term access rate for the UE to be at least $\gamma^*$. Since Algorithm 1 allocates ${\gamma^* f(i,j)/R_1}$ fraction of total time for serving a backhaul link with destination $(i,j)$ and this time is equally divided amongst $f(i,j)$ UEs, the  long term rate for any UE amongst the $f(i,j)$ UEs served on this link equals $\gamma^*$. Similarly, the fraction of time any user gets for access is at least $\gamma^*/R_a$, which implies long term rate of $\gamma^*$. Thus, the upper bound $\gamma^*$ is achievable and the max-min rate is given by $\gamma^*$ if the routing is NNHR.  

Consider any other routing strategy wherein the fiber site activates only nearest neighbour backhaul links. Note that inequality in \eqref{eq:inequMBS} still needs to be satisfied as $f(0,0)$ is independent of the routing. Thus, if the nearest neighbour highway routes are changed such that the new links added to the routes do not directly connect with the MBS, it does not change the max-min rates. The only way $\gamma^*$ is not global optimal is if a non-nearest neighbour backhaul links is activated by the fiber site and it outperforms NNHR. If possible, let the MBS serve some of the traffic on links that are not just limited to ring 1 relays. The new equal rate to all UEs, $\tilde{\gamma}$, has to satisfy the following inequality. 
$\tilde{\gamma}\left(\frac{w_{0,0}}{R_a} + \frac{\beta_1\left(f(0,0) - w_{0,0}\right) }{R_1} + \frac{\beta_2\left(f(0,0) - w_{0,0}\right) }{R_2} +\ldots +\frac{\beta_k\left(f(0,0) - w_{0,0}\right) }{R_k}\right)\leq 1,$
where $\sum_{q = 1}^{k}\beta_q = 1$ and $\beta_1<1$.  Since $R_1>R_2>\ldots>R_k$, we have that $\tilde{\gamma}$ will always be less than that obtained with NNR. Similarly, modifying any of the inequalities in \eqref{eq:inequrelay} to serve some traffic on links with rates $<R_1$ leads to smaller max-min rates as compared to $\gamma^*$. 
\begin{algorithm}[t]
\caption{Theorem~\ref{thm:IAB1} scheduler}\label{algo1}
\begin{algorithmic}[1]
\State Denote $\mathcal{S}_r$ as the set of BSs in ring $r$, where $r = 0,\ldots, k$. Ring $r$ implying distance to $(0,0)$ is $r\mathrm{D}$. Denote by $|\mathcal{S}_r|$ as the cardinality of the set. Total scheduling time is $1$ unit.
		\State The MBS reserves $\frac{w_{0,0} \gamma^*}{R_a}$ fraction of time for access and rest for backhaul. 	
		\State The MBS equally divides the access (backhaul) time frame amongst respective users  that need to be served over access (backhaul) links.
\For{$r=1:k$}
	\For{$q=1:|\mathcal{S}_r|$}
		\State Let $(i,j)$ be the BS indexed by $q$. The BS listens to its parent for backhaul for 
		\Statex \hspace{1.8cm} $\frac{f(i,j)\gamma^*}{R_1}$ fraction of time. This is reserved by its parent already in previous for-loop
		\Statex \hspace{1.8cm} iteration over $r$. Whenever the BS at $(i,j)$ is not listening, it reserves $\frac{w_{i,j}\gamma^*}{R_a}$ 
		\Statex \hspace{1.8cm}  fraction of time for serving access and $\frac{\gamma^* (f(i,j)-w_{i,j})}{R_1}$ for transmitting on backhaul
		\Statex \hspace{1.8cm}  links away from the MBS. In the remaining time, which is non-negative, it stays
		\Statex \hspace{1.8cm} silent. The BS equally divides the access (backhaul) time frame amongst 
		\Statex \hspace{1.8cm} respective users  that need to be served over access (backhaul) links.
	\EndFor	
\EndFor
\end{algorithmic}
\end{algorithm}
%\section{A Lemma Useful For Proving Theorem~\ref{thm:IAB1}.}
%\begin{lem}
%\label{lem:1}
%If $w_{i,j} = w_{-i,-j}, \forall i,j\in\{-k,-k+1,\ldots,k\}$ then $f(i,j) = f(-i,-j)$ with NNHR.
%\end{lem}
%\begin{proof}
%If $(r,s)\in\mathcal{C}(i,j)$, this implies that $(i,j)$ is on the path from $(0,0)\to(r,s)$. Since the definition of highway routing depends only on absolute values of $r$ and $s$, the path from $(0,0)\to(-r,-s)$ is a mirror image of the path from $(0,0)\to(r,s)$. Thus, with NNR if  $(r,s)\in\mathcal{C}(i,j)$, then $(-r,-s)\in\mathcal{C}(-i,-j)$. Thus, $f(-i,-j) =  \sum_{(-r,-s)\in\mathcal{C}(-i,-j)} w_{-r,-s} = \sum_{(r,s)\in\mathcal{C}(i,j)} w_{r,s}$. 
%\end{proof}

%\section*{Acknowledgments}
%The authors thank Gustavo de Veciana, Matthew Andrews, Eugene Visotsky, Mark Cudak for helpful discussions in the course of this work. The authors thank Manan Gupta for help with Fig. 2. 
\bibliographystyle{IEEEtran}
\bibliography{IEEEabrv,Kulkarni}
\end{document}